\documentclass[onecolumn, a4size, 12pt]{IEEEtran}
\usepackage{amsmath}
\usepackage{graphicx}
\usepackage{epsfig}
\usepackage{latexsym}
\usepackage{amsfonts}
\usepackage{amssymb}
\usepackage{psfrag}
\usepackage{cite}
\usepackage{subfigure}
\usepackage{color}
\usepackage{float}
\begin{document}
\pagestyle{plain}

\title{Joint Transmitter and Receiver Energy Minimization in Multiuser OFDM Systems
\thanks{S. Luo and T. J. Lim are with the Department of
Electrical and Computer Engineering, National University of
Singapore (e-mail:\{shixin.luo, eleltj\}@nus.edu.sg).} \thanks{R.
Zhang is with the Department of Electrical and Computer Engineering,
National University of Singapore (e-mail:elezhang@nus.edu.sg). He is
also with the Institute for Infocomm Research, A*STAR, Singapore.}}

\author{Shixin Luo, Rui
Zhang, and Teng Joon Lim}

\maketitle

\begin{abstract}
In this paper, we formulate and solve a weighted-sum transmitter and receiver energy minimization (WSTREMin) problem in the downlink of an orthogonal frequency division multiplexing (OFDM) based multiuser wireless system. The proposed approach offers the flexibility of assigning different levels of importance to base station (BS) and mobile terminal (MT) power consumption, corresponding to the BS being connected to the grid and the MT relying on batteries. To obtain insights into the characteristics of the problem, we first consider two extreme cases separately, i.e., weighted-sum receiver-side energy minimization (WSREMin) for MTs and transmitter-side energy minimization (TEMin) for the BS. It is shown that Dynamic TDMA (D-TDMA), where MTs are scheduled for single-user OFDM transmissions over orthogonal time slots, is the optimal transmission strategy for WSREMin at MTs, while OFDMA is optimal for TEMin at the BS. As a hybrid of the two extreme cases, we further propose a new multiple access scheme, i.e., Time-Slotted OFDMA (TS-OFDMA) scheme, in which MTs are grouped into orthogonal time slots with OFDMA applied to users assigned within the same slot. TS-OFDMA can be shown to include both D-TDMA and OFDMA as special cases. Numerical results confirm that the proposed schemes enable a flexible range of energy consumption tradeoffs between the BS and MTs.
\end{abstract}


\IEEEpeerreviewmaketitle
\setlength{\baselineskip}{1.3\baselineskip}
\newtheorem{definition}{\underline{Definition}}[section]
\newtheorem{fact}{Fact}
\newtheorem{assumption}{Assumption}
\newtheorem{theorem}{\underline{Theorem}}[section]
\newtheorem{lemma}{\underline{Lemma}}[section]
\newtheorem{corollary}{Corollary}
\newtheorem{proposition}{\underline{Proposition}}[section]
\newtheorem{example}{\underline{Example}}[section]
\newtheorem{remark}{\underline{Remark}}[section]
\newtheorem{algorithm}{\underline{Algorithm}}[section]
\newcommand{\mv}[1]{\mbox{\boldmath{$ #1 $}}}

\section{Introduction}\label{sec:introduction}
The range of mobile services available to consumers and businesses is growing rapidly, along with the range of devices used to access these services. Such heterogeneity in both hardware and traffic requirements requires maximum flexibility in all layers of the protocol stack, starting with the physical layer (PHY). Orthogonal frequency division multiple access (OFDMA), which is based on multi-carrier transmission and enables low-complexity equalization of the inter-symbol interference (ISI) caused by frequency selective channels, is one promising PHY solution and has been adopted in various wireless communication standards, e.g., WiMAX and 3GPP LTE \cite{3GPP}. However, the complexity of OFDMA and other features that enable heterogeneous high-rate services leads to increased energy consumption, and hence increased greenhouse gas emissions and operational expenditure. Green radio (GR), which emphasizes improvement in energy efficiency (EE) in bits/joule rather than spectral efficiency (SE) in bits/sec/Hz in wireless networks, has thus become increasingly important and has attracted widespread interest recently \cite{tutorial}.

Prior to the relatively recent emphasis on EE, the research on OFDMA based wireless networks has mainly focused on dynamic resource allocation, which includes dynamic subcarrier (SC) and power allocation, and/or data rate adaptation, for the purposes of either maximizing the throughput \cite{YinLiu00, Song05, KBLee03, Varaiya03} or minimizing the transmit power \cite{Wong99, Cioffi06}. The authors in \cite{Wong99} first considered the problem of power minimization in OFDMA, through adaptive SC and power allocation, subject to transmit power and MTs' individual rate constraints. A time sharing factor, taking values within the interval $[0,1]$, was introduced to relax the original problem to a convex problem, which can then be efficiently solved. The throughput maximization problem for OFDMA can be more generally formulated as a utility maximization problem \cite{Song05}. For example, if the utility function is the network sum-throughput itself, then the maximum value is achieved with each SC being assigned to the MT with the largest channel gain together with the water-filling power allocation over SCs \cite{KBLee03}. This work has been extended to the case of rate proportional fair scheduling in \cite{Varaiya03, Andrew05}. The Lagrange dual decomposition method \cite{Boyd2} was proposed in \cite{Cioffi06} to provide an efficient algorithm for solving OFDMA based resource allocation problems. Although there has been no proof yet for the optimality of the solution by the dual decomposition method, it was shown in \cite{Cioffi06} that with a practical number of SCs, the duality gap is virtually zero.

Recently, there has been an upsurge of interest in EE optimization for OFDMA based networks \cite{Bormann08,Xiong11, Xiong12,Fettweis10,Jianhua13}. Since energy scarcity is more severe at mobile terminals (MTs), due to the limited capacity of batteries, energy-efficient design for OFDMA networks was first considered under the uplink setup \cite{Bormann08}. The sum of MTs' individual EEs, each defined as the ratio of the achievable rate to the corresponding MT's power consumption, is maximized considering both the circuit and transmit power (termed the total power consumption in the sequel). EE maximization for OFDMA downlink transmissions has been studied in \cite{Xiong11, Xiong12,Jianhua13,Fettweis10}. A generalized EE, i.e., the weighted-sum rate divided by the total power consumption, was maximized in \cite{Xiong12} under prescribed user rate constraints. Instead of modeling circuit power as a constant, the authors in \cite{Xiong11, Fettweis10} proposed a model of rate-dependent circuit power, in the context of EE maximization, since larger circuit power is generally required to support a higher data rate.

It is worth noting that most of the existing work on EE-based resource allocation for OFDMA has only considered transmitter-side energy consumption. However, in an OFDMA downlink, energy consumption at the receivers of MTs is also an important issue given the limited power supply of MTs. Therefore, it is interesting to design resource allocation schemes that prolong the operation time of MTs by minimizing their energy usage. Since the energy consumption at the receivers is roughly independent of the data rate and merely dependent on the active time of the MT \cite{Veciana10}, the dominant circuit power consumption at MTs should be considered. Consequently, fast transmission is more beneficial for reducing the circuit energy consumption at the receivers. A similar idea has also been employed in a recent work \cite{Fettweis12}.

In this paper, we propose to characterize the tradeoffs in minimizing the BS's versus MTs' energy consumption in multiuser OFDM based downlink transmission by investigating a weighted-sum transmitter and receiver joint energy minimization (WSTREMin) problem, subject to the given transmission power constraint at the BS and data requirements of individual MTs. We assume that each SC can only be allocated to one MT at each time, but can be shared among different MTs over time, a channel allocation scheme that we refer to as SC time sharing. Therefore, optimal transmission scheduling at the BS involves determining the time sharing factors and the transmit power allocations over the SCs for all MTs.

To obtain useful insights into the optimal energy consumption for the BS and MTs, we first consider two extreme cases separately, i.e., the weighted-sum receiver-side energy minimization (WSREMin) for MTs and transmitter-side energy minimization (TEMin) for the BS. It is shown that Dynamic TDMA (D-TDMA) as illustrated in Fig. \ref{fig:Schemeexample}(a), where MTs are scheduled in orthogonal time slots for transmission, is the optimal strategy for WSREMin at MTs. Intuitively, this is because D-TDMA minimizes the receiving time of individual MTs given their data requirements. In contrast, OFDMA as shown in Fig. \ref{fig:Schemeexample}(b) is proven to be optimal for TEMin at the BS. It is observed that transmitter-side energy and weighted-sum receive energy consumptions cannot be minimized at the same time in general due to different optimal transmission schemes, and there exists a tradeoff between the energy consumption of the BS and MTs. To obtain more flexible energy consumption tradeoffs between the BS and MTs for WSTREMin and inspired by the results from the two extreme cases, we further propose a new multiple access scheme, i.e., Time-Slotted OFDMA (TS-OFDMA) scheme as illustrated in Fig. \ref{fig:Schemeexample}(c), in which MTs are grouped into orthogonal time slots with OFDMA applied when multiple users are assigned to the same time slot. TS-OFDMA can be shown to include both the D-TDMA and OFDMA as special cases.

\begin{figure}
\centering
\epsfxsize=0.7\linewidth
    \includegraphics[width=9cm]{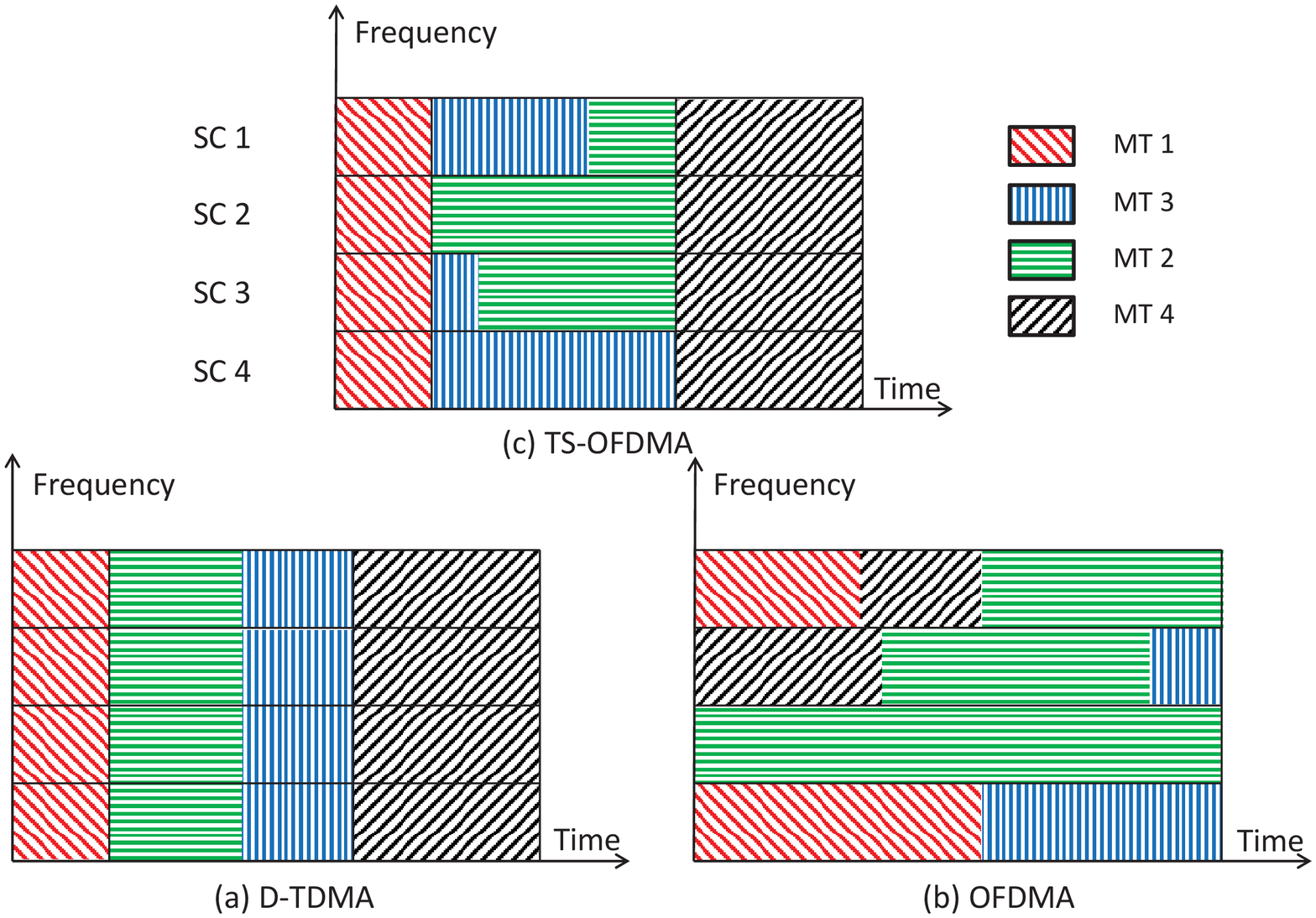}
\caption{Transmission schemes: (a) Dynamic TDMA (D-TDMA); (b) OFDMA; and (c) Time-Slotted OFDMA (TS-OFDMA).} \label{fig:Schemeexample}
\end{figure}

The rest of this paper is organized as follows. Section \ref{sec:system model} introduces the multiuser OFDM based downlink system model, and the power consumption models for the BS and MTs. Section \ref{sec:WSRE} and Section \ref{sec:TE} then study the two extreme cases of WSREMin and TEMin, respectively. Section \ref{sec:WSE} introduces the general WSTREMin problem and proposes the TS-OFDMA transmission scheme to achieve various energy consumption tradeoffs between the BS and MTs. {\color{red}In Section \ref{sec:delay}, we discuss how the obtained results can be extended to the case when a maximum time constraint is imposed on the transmission.} Section \ref{sec:numerical} shows numerical results. Finally Section \ref{sec:conclusion} concludes the paper.

\section{System Model and Problem Formulation}\label{sec:system model}
\subsection{System Model}
Consider a multiuser OFDM-based downlink transmission system consisting of one BS, $N$ orthogonal subcarriers (SCs) each with a bandwidth of $W$ Hz, and $K$ MTs. Let $\mathcal{K}$ and $\mathcal{N}$ denote the sets of MTs and SCs, respectively. We assume that each SC can be assigned to at most one MT at any given time, but the SC assignment is allowed to be shared among MTs over time, i.e., SC time sharing. We also assume that the noise at the receiver of each MT is modeled by an additive white Gaussian noise (AWGN) with one-sided power spectrum density denoted by $N_0$. Let $p_{k,n}$ be the transmit power allocated to MT $k$ in SC $n$, $k \in \mathcal{K}$, $n \in \mathcal{N}$, and $r_{k,n}$ be the achievable rate of MT $k$ at SC $n$ in the downlink. Then it follows that
\begin{align}\label{eq:rate on each subcarrier}
r_{k,n} = W\log_2\left(1 + \frac{h_{k,n}p_{k,n}}{\Gamma N_0W}\right)
\end{align}
where $\Gamma \geq 1$ accounts for the gap from the channel capacity due to practical modulation and coding, and $h_{k,n}$ is the channel power gain from the BS to MT $k$ at SC $n$, which is assumed to be perfectly known at both the BS and MT $k$.

With time sharing of SCs among MTs, $\rho_{k,n}$, dubbed the time sharing factor, is introduced to represent the fraction of time that SC $n$ is assigned to MT $k$, where $0 \leq \rho_{k,n} \leq 1$, $\forall k, n$ and $\sum^{K}_{k=1}\rho_{k,n} \leq 1$, $\forall n$. Let $T$ denote the total transmission time for our proposed scheduling. The amount of information bits delivered to MT $k$ over time $T$ is thus given by
\begin{align}\label{eq:received data by each user}
Q_k = T\sum^{N}_{n=1}\rho_{k,n}r_{k,n}.
\end{align}
The average transmit power is given by
\begin{align}\label{eq:average power}
\bar{P} = \sum^{K}_{k=1}\sum^{N}_{n=1}\rho_{k,n}p_{k,n}.
\end{align}
We assume that $\bar{Q}_k$ bits of data need to be delivered from the BS to MT $k$ over a slot duration $T$ for the time slot of interest. Then the following constraint must be satisfied:
\begin{align}\label{eq:QoS constraint for each user}
Q_k \geq \bar{Q}_k, \forall k \in {\cal K}.
\end{align}

We further assume that the receiver of each MT is turned on only when the BS starts to send the data it requires, which can be at any time within the time slot, and that it is turned off right after all $\bar{Q}_k$ bits of data are received. Let $t_k$, $0 \leq t_k\leq T$, denote the ``on'' period of MT $k$. It is observed that the following inequalities must hold for all MTs:
\begin{align}\label{eq:lower bound of on time}
\max\limits_n \{T\rho_{k,n}\} \leq t_k \leq T, \forall k \in {\cal K}.
\end{align}
The origin of this inequality can be understood from Fig. \ref{fig:SystemModel}, where MT $k$ is turned on and then off within the time interval $T$.
\begin{figure}
\centering
\epsfxsize=0.7\linewidth
    \includegraphics[width=9cm]{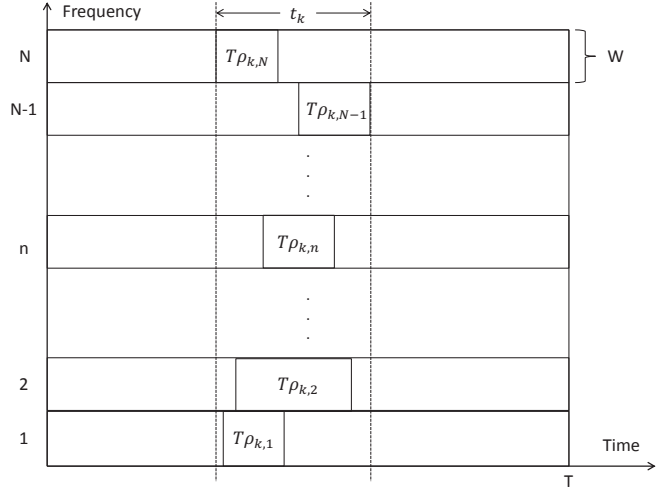}
\caption{Multiuser OFDM transmission with SC time sharing.} \label{fig:SystemModel}
\end{figure}

Energy consumption at the BS in general comprises two major parts: transmit power $\bar{P}$ and a constant power $P_{t,c}$ accounting for all non-transmission related energy consumption due to e.g. processing circuits and cooling. Consequently, the total energy consumed by the BS over duration $T$, denoted by $E_t$, can be modeled as
\begin{align}\label{eq:transmitter-side energy}
E_t = T\bar{P} + TP_{t,c}.
\end{align}

On the other hand, the power consumption at the receiver of each MT is assumed to be constant \cite{Veciana10}, denoted by $P_{r,c}$, when it is in the ``on'' period receiving data from the BS. Otherwise, if the receiver does not receive any data from the BS, its consumed power is in general negligibly small and thus is assumed to be zero. Hence, the receiver energy consumed by each MT $k$ over $T$, denoted by $E_{r,k}$, can be approximately modeled as
\begin{align}\label{eq:receiver energy at MT k}
E_{r,k} = P_{r,c}t_k, k \in {\cal K}.
\end{align}
In general, each MT can be in a different state of energy depletion, and thus it is sensible to define a weighted-sum receiver-side energy (WSRE) consumption of all MTs as
\begin{align}\label{eq:weighted energy definition}
E^w_r = \sum^{K}_{k=1}\alpha_kE_{r,k}
\end{align}
where a larger weight $\alpha_k$ reflects the higher priority of MT $k$ in terms of energy minimization.

{\color{red}It is assumed that all channels $h_{k,n}$'s are constant over the total transmission time of a frame, $T$. While in theory the optimal $T$ is unbounded, for a practical number of bits to be transmitted per frame, $\bar{Q}_k$'s, and practical transmit power levels $P_{t,c}$ and $P_{r,c}$, the designed optimal $T$ will be finite and in fact usually quite small. If we consider low-mobility and/or short frame lengths, then the assumption of a static channel over an indeterminate $T$ is valid. However, in Section \ref{sec:delay}, we provide detailed discussions on how the obtained results in this paper can be extended to the case when an explicit maximum transmission time constraint is imposed.}

%We want to emphasize here that a low mobility environment is assumed in this work, i.e., the considered group of $K$ MTs move slowly with respect to the BS, such that the channel gains from the BS to each MT change slowly over time. Furthermore, as energy consumption being the major concern which grows linearly with the total transmission time $T$ for both the BS (due to non-transmission related energy consumption) and MTs, with practical data package size for one frame transmission, i.e., $\bar{Q}_k$'s, and energy consumption model, i.e., $P_{t,c}$ and $P_{r,c}$'s, $T$ is always a bounded value. Therefore, in this paper, we assume that $h_{k,n}$'s remain unchanged within the transmission time of interest, and there is no maximum constraint imposed on $T$. We make the above assumptions due to the following two main reasons: first, it allows us to derive the most energy-efficient multiple access scheme for the MTs; second, it makes our analysis more mathematically tractable. Therefore, this paper is mainly a theoretic study on characterizing the performance limit of energy-efficient transmission in multiuser OFDM based downlink systems.

\subsection{Problem formulation}
We aim to characterize the tradeoffs in minimizing the BS's versus MTs' energy consumption, i.e., $E_t$ versus $E_{r,k}$'s, in multiuser OFDM based downlink transmission by investigating a weighted-sum transmitter and receiver joint energy minimization (WSTREMin) problem, which is formulated as
\begin{align}
&(\text{WSTREMin}): \nonumber \\~&\mathop{\mathtt{Min.}}\limits_{\{p_{k,n}\}, \{\rho_{k,n}\},T}
~~ \sum^{K}_{k=1}\alpha_kt_kP_{r,c} \nonumber \\
&+ \alpha_0\left(\sum^{K}_{k=1}\sum^{N}_{n=1} T\rho_{k,n}p_{k,n} + TP_{t,c}\right)\\
\mathtt{s.t.}
& ~~ \sum^{K}_{k=1}\rho_{k,n} \leq 1, \forall n \label{eq:WSTREMin c1}\\
& ~~ \sum^{N}_{n=1} T\rho_{k,n}r_{k,n} \geq \bar{Q}_k, \forall k \label{eq:WSTREMin c2}\\
& ~~ \sum^{K}_{k=1}\sum^{N}_{n=1} \rho_{k,n}p_{k,n} \leq P_{\text{avg}} \label{eq:WSTREMin c3}\\
& ~~ T > 0, p_{k,n} \geq 0, ~ 0 \leq \rho_{k,n} \leq 1, \forall n, k \label{eq:WSTREMin c4}
\end{align}
where $\alpha_0$ is an additional weight assigned to the BS, which controls the resulting minimum energy consumption of the BS as compared to those of MTs. Notice that the design variables in the above problem include the power allocation $p_{k,n}$, time sharing factor $\rho_{k,n}$, as well as transmission time $T$, while the constraints in (\ref{eq:WSTREMin c1}) are to limit the total transmission time at each SC to be within $T$, those in (\ref{eq:WSTREMin c2}) are for the data requirements of different MTs, and that in (\ref{eq:WSTREMin c3}) specifies the average transmit power at BS, denoted by $P_{\text{avg}}$. The main difficulty in solving problem (WSTREMin) lies in the absence of a functional relationship among $t_k$, $\rho_{k,n}$'s and $T$ with the inequality in (\ref{eq:lower bound of on time}) being the only known expression that links the three variables. Minimizing over the upper bound of each MT's energy consumption, i.e., $TP_{r,c}$, which could be quite loose as illustrated in Fig. \ref{fig:SystemModel}, may result in conservative or energy-inefficient solution.

In order to obtain useful insights into the optimal energy consumption for the BS and MTs, we first consider two extreme cases separately in the following two sections, i.e., WSRE minimization (WSREMin) corresponding to the case of $\alpha_0 = 0$ in Section \ref{sec:WSRE} and transmitter-side energy minimization (TEMin) corresponding to the case of $\alpha_k = 0, \forall k$, respectively, in Section \ref{sec:TE}. Compared with problem (WSTREMin), problems (WSREMin) and (TEMin) have exactly the same set of constraints but different objective functions. We will illustrate how problem (WSTREMin) may be practically solved based on the results from the the two extreme cases in Section \ref{sec:WSE}.

\begin{remark}
Problem (WSTREMin) could have an alternative interpretation by properly setting the energy consumption weights $\alpha_0$ and $\alpha_k$'s. Suppose $\alpha_0$ and $\alpha_k$ represent the unit cost of consumed energy at the BS and MT $k$, respectively. Since MTs are usually powered by capacity limited batteries in comparison to the electrical grid powered BS, $\alpha_0$ and $\alpha_k$'s should reflect the energy price in the market for the BS and the risk of running out of energy for each MT $k$, respectively. With this definition, problem (WSTREMin) can be treated as a network-wide cost minimization problem. How to practically select the values of $\alpha_0$ and $\alpha_k$'s to achieve this end is beyond the scope of this paper.
\end{remark}

\section{Receiver-Side Energy Minimization}\label{sec:WSRE}
In this section, we consider minimizing receiver energy consumption at all MTs without regard for BS energy consumption. From (\ref{eq:receiver energy at MT k}) and (\ref{eq:weighted energy definition}), the WSREMin problem is thus formulated as
\begin{align}
&(\text{WSREMin}): \nonumber \\~&\mathop{\mathtt{Min.}}\limits_{\{p_{k,n}\}, \{\rho_{k,n}\}, T}
~~ \sum^{K}_{k=1}\alpha_kP_{r,c}t_k  \\
\mathtt{s.t.}
& ~~ (\ref{eq:WSTREMin c1}), (\ref{eq:WSTREMin c2}), (\ref{eq:WSTREMin c3}), ~ \text{and} ~(\ref{eq:WSTREMin c4}).
\end{align}
As mentioned in Section \ref{sec:introduction}, receiver-side energy minimization has also been considered in \cite{Fettweis12}, in which the available time-frequency resources are divided into equally spaced RBs over both time and frequency. Flat-fading, i.e., the channels are the same across all the RBs, was assumed for each MT, based on which an integer programme with each MT constrained by the number of required RBs is formulated for RB allocation. Problem (WSREMin), in contrast, assumes a more flexible SC allocation with time sharing factor $\rho_{k,n}$'s to achieve further energy saving. Moreover, the optimal power allocation corresponding to frequency selective channels is obtained.

Similar to problem (WSTREMin), the main difficulty in solving problem (WSREMin) lies in the absence of a functional relationship among $t_k$, $\rho_{k,n}$'s and $T$. However, it can be shown that a dynamic TDMA (D-TDMA) based solution, i.e., MTs are scheduled for single-user OFDM transmission over orthogonal slots with respective duration $\rho_kT$, $k=1,\cdots,K$, with $\sum^{K}_{k=1}\rho_{k} \leq 1$, is optimal for problem (WSREMin), as given in the following proposition.
\begin{proposition}\label{proposition:1}
Let $\rho^{*}_{k,n}$, $n=1,\cdots,N$, and $t^{*}_k$ denote the optimal set of time sharing factors and the optimal ``on'' period for MT $k$, respectively, $k \in \mathcal{K}$, in problem (WSREMin). Then, we have
\begin{align}
\rho^{*}_{k,n} & = \rho^{*}_{k}, ~~ \forall n \in \mathcal{N}, k \in \mathcal{K} \\
t^{*}_k & = T\rho^{*}_{k}, \forall k \in \mathcal{K} \label{eq:TDMA statement}
\end{align}
where $\rho^{*}_{k}$ denotes the common value of all $\rho^{*}_{k,n}$, $\forall n \in \mathcal{N}$, for MT $k$.
\end{proposition}
\begin{proof}
See Appendix \ref{appendix:proof proposition 1}.
\end{proof}

\begin{remark}
Proposition \ref{proposition:1} indicates that the time sharing factors at all SCs should be identical for each MT $k$ to minimize its ``on'' period, which is achieved by D-TDMA transmission as shown in Fig. \ref{fig:Schemeexample}(a). Notice that D-TDMA minimizes the on time of each MT and therefore their weighted energy consumption, as will be shown next. However, it extends the transmission time of BS, $T$, and thus may not be optimal from the viewpoint of BS energy saving, as we shall see in Section \ref{sec:TE}.
\end{remark}

With Proposition \ref{proposition:1} and $t_k$'s given in (\ref{eq:TDMA statement}), the WSREMin problem under D-TDMA is formulated as
\begin{align}
&(\text{WSREMin-TDMA}): \nonumber \\~&\mathop{\mathtt{Min.}}\limits_{\{p_{k,n} \geq 0\}, \{t_k > 0\}}
~~ \sum^{K}_{k=1}\alpha_kP_{r,c}t_k  \\
\mathtt{s.t.}
& ~~ \sum^{N}_{n=1} t_kr_{k,n} \geq \bar{Q}_k, \forall k \label{eq:WSRT-TDMA c1} \\
& ~~ \sum^{K}_{k=1}\sum^{N}_{n=1}t_kp_{k,n} \leq P_{\text{avg}}\sum^{K}_{k=1}t_k. \label{eq:WSRT-TDMA c2}
\end{align}
It is observed that problem (WSREMin-TDMA) is non-convex due to the coupled terms $t_kr_{k,n}$ in (\ref{eq:WSRT-TDMA c1}) and $t_kp_{k,n}$ in (\ref{eq:WSRT-TDMA c2}). By a change of variables $s_{k,n} = t_kr_{k,n}$, $\forall k, n$, problem (WSREMin-TDMA) can be reformulated as
\begin{align}
\mathrm{(P1)}:~\mathop{\mathtt{Min.}}\limits_{\{s_{k,n} \geq 0\}, \{t_k > 0\}} &
~~ \sum^{K}_{k=1}\alpha_kP_{r,c}t_k   \label{eq:P1 objective}\\
\mathtt{s.t.}
& ~~ \sum^{N}_{n=1} s_{k,n} \geq \bar{Q}_k, \forall k \label{eq:P1 c1}\\
& ~~ \sum^{K}_{k=1}\sum^{N}_{n=1}t_k\frac{e^{a\frac{s_{k,n}}{t_k}}-1}{f_{k,n}} \leq  P_{\text{avg}}\sum^{K}_{k=1}t_k \label{eq:P1 c2}
\end{align}
where $f_{k,n} = \frac{h_{k,n}}{\Gamma N_0W}$ and $a = \frac{\ln2}{W}$. Note that the objective function in (\ref{eq:P1 objective}) and constraints in (\ref{eq:P1 c1}) are all affine, while the constraints in (\ref{eq:P1 c2}) are convex due to the fact that the function $t_ke^{a\frac{s_{k,n}}{t_k}}$ is the perspective of a strictly convex function $e^{as_{k,n}}$ with $a > 0$, and thus is a convex function \cite{Boydbook}. As a result, problem (P1) is convex. Thus, the Lagrange duality method can be applied to solve this problem exactly \cite{Boydbook}.

In the rest of this section, instead of solving the dual of problem (P1) directly which involves only numerical calculation and provides no insights, we develop a simple bisection search algorithm by revealing the structure of the optimal solution to problem (WSREMin-TDMA), given in the following theorem.

\begin{theorem}\label{theorem:1}
Let $\boldsymbol\lambda^{*}= [\lambda^{*}_1,\cdots,\lambda^{*}_K] \geq \mathbf{0}$ and $\beta^{*} \geq 0$ denote the optimal dual solution to problem (P1). The optimal solution of problem (WSREMin-TDMA) is given by
\begin{align}
p^{*}_{k,n} & = \left(\frac{\lambda^{*}_k}{a\beta^{*}} - \frac{1}{f_{k,n}}\right)^{+} \label{eq:P1 sol1}\\
t^{*}_k & = \frac{a\bar{Q}_k}{\sum^{N}_{n=1}\left(\ln\frac{\lambda^{*}_kf_{k,n}}{a\beta^{*}}\right)^{+}} \label{eq:P1 sol2}
\end{align}
where $\boldsymbol\lambda^{*}$ and $\beta^{*}$ need to satisfy
\begin{align}
\beta^{*} - \min\limits_k(\alpha_k)P_{r,c}/P_{\text{avg}} &< 0 \\
\alpha_kP_{r,c} - \beta^{*} P_{\text{avg}} + \sum^{N}_{n=1}u_{n}(\beta^{*},\lambda^{*}_k)  & = 0, \forall k \in \mathcal{K} \label{eq:P1 sol3}
\end{align}
where $u_{n}(\beta,\lambda_k) = \left(\frac{\lambda_k}{a} - \frac{\beta}{f_{k,n}}\right)^{+} - \frac{\lambda_k}{a}\left(\ln\frac{\lambda_kf_{k,n}}{a\beta}\right)^{+}$ and $(\cdot)^{+} \triangleq \max\{\cdot,0\}$.
\end{theorem}
\begin{proof}
See Appendix \ref{appendix:proof theorem 1}.
\end{proof}
It is observed from (\ref{eq:P1 sol1}) that the optimal power allocation has a water-filling structure \cite{Tse}, except that the water levels are different over MTs. These are specified by $\lambda^{*}_k$ for MT $k$ and need to be found by solving the equations in (\ref{eq:P1 sol3}).
Since it can be shown that $\sum^{N}_{n=1}u_{n}(\beta,\lambda_k) \leq 0$ is strictly decreasing in $\lambda_k$ given $\beta < \min\limits_{k}\{\alpha_k\}P_{r,c}/P_{\text{avg}}$, with the assumption of identical channels for all the MTs, it is observed that larger $\alpha_k$ results in larger $\lambda^{*}_k$ or higher water-level, which means more power should be allocated to the MT that has higher priority in terms of energy minimization.

Based on Theorem \ref{theorem:1}, one algorithm to solve problem (WSREMin-TDMA) is given in Table \ref{table1}, in which $\beta^{*}$ is obtained through bisectional search until the average power constraint in (\ref{eq:WSRT-TDMA c2}) is met with equality. For the algorithm given in Table \ref{table1}, the computation time is dominated by updating the power and time allocation with given $\beta$ in steps b)-d), which is of order $KN$. Since the number of iterations required for the bisection search over $\beta$ is independent of $K$ and $N$, the overall complexity of the algorithm in Table \ref{table1} is $\mathcal{O}(KN)$.

\begin{table}[ht]
\begin{center}
\caption{\textbf{Algorithm 1}: Algorithm for Solving Problem (WSREMin-TDMA)} \vspace{0.2cm}
 \hrule
\vspace{0.3cm}
\begin{enumerate}
\item {\bf Given} $\beta_{\text{min}}(\triangleq 0) \leq \beta^{*} < \beta_{\text{max}} (\triangleq \min\limits_k(\alpha_k)P_{r,c}/P_{\text{avg}}$).
\item {\bf Repeat}
    \begin{itemize}
        \item[ a)] $\beta = \frac{1}{2}\left(\beta_{\text{min}} + \beta_{\text{max}}\right)$.
        \item[ b)] Obtain $\lambda_k$ such that $u(\beta,\lambda_k) = 0$, where $u(\beta,\lambda_k) = \alpha_kP_{r,c} - \beta P_{\text{avg}} + \sum^{N}_{n=1}u_{n}(\beta,\lambda_k)$, $k = 1, \cdots, K$,.
        \item[ c)] Obtain $p_{k,n}$ and $t_k$ according to (\ref{eq:P1 sol1}) and (\ref{eq:P1 sol2}) for $k = 1, \cdots, K$, $n = 1, \cdots, N$.
        \item[ d)] If $\sum^{K}_{k=1}\sum^{N}_{n=1}t_kp_{k,n} \geq P_{\text{avg}}\sum^{K}_{k=1}t_k$, set $\beta_{\text{min}} \leftarrow \beta$; otherwise, set $\beta_{\text{max}} \leftarrow \beta$.
    \end{itemize}
\item {\bf Until} $\beta_{\text{max}} - \beta_{\text{min}} < \delta$ where $\delta$ is a small positive constant that controls the algorithm accuracy.
\end{enumerate}
\vspace{0.2cm} \hrule \label{table1} \end{center}
\end{table}

\section{Transmitter-Side energy minimization}\label{sec:TE}
In this section, we study the case of minimizing the energy consumption at the BS while ignoring the receiver energy consumption at MTs. From (\ref{eq:average power}) and (\ref{eq:transmitter-side energy}), the transmitter-side energy minimization (TEMin) problem is formulated as
\begin{align}
&(\text{TEMin}): \nonumber \\~&\mathop{\mathtt{Min.}}\limits_{\{p_{k,n}\}, \{\rho_{k,n}\},T}
~~ \sum^{K}_{k=1}\sum^{N}_{n=1} T\rho_{k,n}p_{k,n} + TP_{t,c}  \\
\mathtt{s.t.}
& ~~ (\ref{eq:WSTREMin c1}), (\ref{eq:WSTREMin c2}), (\ref{eq:WSTREMin c3}), ~ \text{and} ~(\ref{eq:WSTREMin c4}).
\end{align}
A similar formulation has been considered in \cite{Xiong12,Jianhua13,Fettweis10,Xiong11}, in which the energy efficiency, defined as the ratio of the achievable rate to the total power consumption, is maximized under prescribed user rate constraints. Problem (TEMin), in contrast, considers the data requirements $\bar{Q}_k$'s and includes the transmission time $T$ as a design variable to explicitly address the tradeoffs between the transmission and non-transmission related energy consumption at BS: longer transmission time results in larger non-transmission related energy consumption $TP_{t,c}$ but smaller transmission related energy consumption $\sum^{K}_{k=1}\sum^{N}_{n=1} T\rho_{k,n}p_{k,n}$ with given data requirements \cite{Bormann08}.

Problem (TEMin) is also non-convex due to the coupled terms $T\rho_{k,n}r_{k,n}$ in (\ref{eq:WSTREMin c2}) and $\rho_{k,n}p_{k,n}$ in (\ref{eq:WSTREMin c3}). Compared with \cite{Xiong12,Jianhua13,Fettweis10,Xiong11}, it is observed that the design variable $T$ further complicates the problem. To solve this problem, we propose to decompose problem (TEMin) into two subproblems as follows.
\begin{align}
(\text{TEMin-1}):~\mathop{\mathtt{Min.}}\limits_{\{p_{k,n}\}, \{\rho_{k,n}\}} &
~~ \sum^{K}_{k=1}\sum^{N}_{n=1} \rho_{k,n}p_{k,n} \\
\mathtt{s.t.}
& ~~ (\ref{eq:WSTREMin c2}) ~ \text{and} ~ (\ref{eq:WSTREMin c3}) \\
& ~~ p_{k,n} \geq 0, ~ 0 \leq \rho_{k,n} \leq 1, \forall n, k. \label{eq:TE1 c3}
\end{align}
\begin{align}
(\text{TEMin-2}):~\mathop{\mathtt{Min.}}\limits_{T} &
~~ Tv(T) + TP_{t,c}  \label{eq:TE2 objective}\\
\mathtt{s.t.}
& ~~ v(T) \leq P_{\text{avg}} \label{eq:TE c3} \\
& ~~ T > 0.
\end{align}
Here, $v(T)$ denotes the optimal value of the objective function in problem (TEMin-1). Note that problem (TEMin-1) minimizes the BS average transmit power with given transmission time $T$ and a set of data constraints $\bar{Q}_k$. Then problem (TEMin-2) searches for the optimal $T$ to minimize the total energy consumption at BS subject to the average transmit power constraint, $P_{\text{avg}}$. In the rest of this section, we first solve problem (TEMin-1) with given $T >0$. Then, we show that problem (TEMin-2) is convex and can be efficiently solved by a bisection search over $T$.

\subsection{Solution to Problem (TEMin-1)}\label{sec:TEMin1}
With given $T>0$, the data requirement $\bar{Q}_k$ for MT $k$ can be equivalently expressed in terms of rate as $c_{k} = \frac{\bar{Q}_k}{T}$. Similarly as for problem (P1), we make a change of variables as $m_{k,n} = \rho_{k,n}r_{k,n}, \forall k, n$. Moreover, we define $\frac{m_{k,n}}{\rho_{k,n}} = 0$ at $m_{k,n} = \rho_{k,n} = 0$ to maintain continuity at this point. Problem (TEMin-1) is then reformulated as
\begin{align}
\mathrm{(P2)}:\mathop{\mathtt{Min.}}\limits_{\{m_{k,n}\}, \{\rho_{k,n}\}} &
~~ \sum^{K}_{k=1}\sum^{N}_{n=1} \rho_{k,n}\frac{e^{a\frac{m_{k,n}}{\rho_{k,n}}}-1}{f_{k,n}}  \label{eq:P2 objective}\\
\mathtt{s.t.}
& ~~ \sum^{K}_{k=1}\rho_{k,n} \leq 1, \forall n \label{eq:P2 c1}\\
& ~~ \sum^{N}_{n=1} m_{k,n} \geq c_k, \forall k \label{eq:P2 c2} \\
& ~~ m_{k,n} \geq 0, 0 \leq \rho_{k,n} \leq 1, \forall k, n.
\end{align}
Although problem (P2) can be shown to be convex just as for problem (P1), it does not have the provably optimal structure for SC allocation given in Proposition \ref{proposition:1}. In this case, in general the SC's are shared among all MTs at any given time, denoted by the set of time sharing factors $\{\rho_{k,n}\}$, which are different for all $k$ and $n$ in general. Since problem (P2) is convex, the Lagrange duality method can be applied to solve this problem optimally. Another byproduct of solving problem (P2) by this method is the corresponding optimal dual solution of problem (P2), which will be shown in the next subsection to be the desired gradient of the objective function in problem (TEMin-2) required for solving this problem. The details of solving problem (P2) and its dual problem through the Lagrange duality method can be found in Appendix \ref{appendix:solution to p2} with one algorithm summarized in Table \ref{table2}.

We point out here that the problem of transmit power minimization for OFDMA downlink transmission with SC time sharing has also been studied in \cite{Varaiya03,Wong99}. In \cite{Varaiya03}, problem (P2) is solved directly without introducing its dual problem, but in this paper, the corresponding dual solution is the gradient of the objective function in problem (TEMin-2) and therefore the dual problem is important. In \cite{Wong99}, the dual variables are updated one at a time until the data rate constraints for all users are satisfied, and this is extremely slow. In this paper, the optimal dual solution of problem (P2) is obtained more efficiently by the ellipsoid method \cite{Boyd2}. Since with the optimal dual solutions, we may obtain infinite sets of primal solution, and some might not satisfy the constraints in (\ref{eq:P2 c1}) and/or (\ref{eq:P2 c2}) \cite{Rui06}, the optimal solution of problem (P2) is further obtained by solving a linear feasibility problem (more details are given in Appendix \ref{appendix:solution to p2}). Finally, in \cite{Varaiya03,Wong99}, the time sharing factor $\rho_{k,n}$ is treated as a relaxed version of the SC allocation indicator, which needs to be quantized to be $0$ or $1$ after solving problem (P2). However, since problem (P2) in this paper is only a subproblem of problem (TEMin), in which the transmission time $T$ is a design variable, SC time sharing can indeed be implemented with proper scheduling at the BS such that each SC is still assigned to at most one MT at any given time.

\subsection{Solution to Problem (TEMin-2)}\label{sec:TEMin2}
With problem (TEMin-1) solved, we proceed to solve problem (TEMin-2) in this subsection. First, we have the following lemma.
\begin{lemma}\label{lemma:4}
Problem (TEMin-2) is convex.
\end{lemma}
\begin{proof}
See Appendix \ref{appendix:proof lemma 4}.
\end{proof}
Since problem (TEMin-2) is convex, and $v(T)$ is continuous and differentiable \cite{MilgromSegal02}, a gradient based method e.g. Newton method \cite{Boydbook} can be applied to solve problem (TEMin-2), where the required gradient is given in the following lemma.
\begin{lemma}\label{lemma:5}
The gradient of $v(T)T+ P_{t,c}T$ with respect to $T$, $T > 0$, is given by
\begin{align}\label{eq:gradient of objective}
v(T) -\frac{1}{T}\sum^{K}_{k=1}\lambda^{*}_k(T)\bar{Q}_k + P_{t,c}
\end{align}
where $\{\lambda^{*}_k(T)\}$ is the optimal dual solution of problem (P2) with given $T > 0$.
\end{lemma}
\begin{proof}
See Appendix \ref{appendix:proof lemma 5}.
\end{proof}

\subsection{Algorithm for problem (TEMin)}
With both problems (TEMin-1) and (TEMin-2) solved, the solution of problem (TEMin) can be obtained by iteratively solving the above two problems. In summary, an algorithm to solve problem (TEMin) is given in Table \ref{table3}. For the algorithm given in Table \ref{table3}, the computation time is dominated by obtaining $v(T)$ and $\boldsymbol\lambda^{*}(T)$ with given $T$ through the algorithm in Table \ref{table2} of Appendix \ref{appendix:solution to p2}, which is of order $K^4+N^4+K^3N^3$. Similarly, since the number of iterations required for the bisection search over $T$ is independent of $K$ and $N$, the overall complexity of the algorithm given in Table \ref{table3} bears the same order over $K$ and $N$ as that for the algorithm in Table \ref{table2} of Appendix \ref{appendix:solution to p2}, which is $\mathcal{O}(K^4+N^4+K^3N^3)$.

\begin{table}[ht]
\begin{center}
\caption{\textbf{Algorithm 3}: Algorithm for Solving Problem (TEMin)} \vspace{0.2cm}
 \hrule
\vspace{0.3cm}
\begin{enumerate}
\item Define $y(T) \triangleq v(T) -\frac{1}{T}\sum^{K}_{k=1}\lambda^{*}_k(T)\bar{Q}_k + P_{t,c}$, where $v(T)$ and $\boldsymbol\lambda^{*}(T)$ are obtained by the Algorithm 2 in Table \ref{table2} of Appendix \ref{appendix:solution to p2}.
\item Obtain $T^{'}$ through bisection search such that $y(T^{'}) = 0$.
\item {\bf If } $v(T^{'}) \leq P_{\text{avg}}$, {\bf then} $T^{*} = T^{'}$; {\bf otherwise} find $T^{*}$ through bisection search such that $v(T^{*}) = P_{\text{avg}}$.
\item Obtain the optimal solution of problem (P2), i.e., $\{\{m^{*}_{k,n}\}, \{\rho^{*}_{k,n}\}\}$, with $T^{*}$ by the Algorithm 2 in Table \ref{table2} of Appendix \ref{appendix:solution to p2}.
\item Obtain the optimal solution of problem (TEMin), i.e., $\{\{p^{*}_{k,n}\}, \{\rho^{*}_{k,n}\}\}$, as $p^{*}_{k,n} = \left(2^{m^{*}_{k,n}/\rho^{*}_{k,n}}-1\right)/f_{k,n}, \forall k,n$.
\end{enumerate}
\vspace{0.2cm} \hrule \label{table3} \end{center}
\end{table}
\begin{remark}
Compared with the D-TDMA based solution in Section \ref{sec:WSRE} for the case of receiver-side energy minimization, the optimal solution of problem (TEMin) for transmitter-side energy minimization implies that OFDMA (c.f. Fig. \ref{fig:Schemeexample}(b)), in which the $N$ SCs are shared among all MTs at any given time, needs to be employed. However, OFDMA may prolong the active time of individual MTs, i.e., $t_k$'s, and is thus not energy efficient in general from the perspective of MT energy saving.
\end{remark}

\section{Joint Transmit and Receive Energy Minimization}\label{sec:WSE}
From the two extreme cases studied in Sections \ref{sec:WSRE} and \ref{sec:TE}, we know that D-TDMA as shown in Fig. \ref{fig:Schemeexample}(a) is the optimal transmission strategy to minimize the weighted-sum receive energy consumption at the MT receivers; however, OFDMA as shown in Fig. \ref{fig:Schemeexample}(b) is optimal to minimize the energy consumption at the BS transmitter. There is evidently no single strategy that can minimize the BS's and MTs' energy consumptions in OFDM-based multiuser downlink transmission. In this section, motivated by the solutions derived from the previous two special cases, we propose a new multiple access scheme termed Time-Slotted OFDMA (TS-OFDMA) transmission scheme, which includes D-TDMA and OFDMA as special cases, and propose an efficient algorithm to approximately solve problem (WSTREMin) using the proposed TS-OFDMA.

\subsection{TS-OFDMA}\label{sec:TS-OFDMA}
The TS-OFDMA scheme is described as follows. The total transmission time $T$ is divided into $J$ orthogonal time slots with $1 \leq J \leq K$. The $K$ MTs are then assigned to each of the $J$ slots for downlink transmission. Let $\Phi_j$ represent the set of MTs assigned to slot $j$, $j=1,\cdots,J$. We thus have
\begin{align}
\Phi_j \cap \Phi_k & = \varnothing, \forall j \neq k \\
\bigcup_j \Phi_j & = \mathcal{K}.
\end{align}
The period that each MT $k$ is switched on (versus off) then equals the duration of its assigned slot, denoted by $T_j$, i.e., $t_k = T_j$ if $k \in \Phi_j$, with $\sum^{J}_{j=1}T_j=T$. Notice that TS-OFDMA includes D-TDMA (if $J=K$) and OFDMA (if $J=1$) as two special cases\footnote{\color{red}Note that OFDMA is considered as a flexible transmission scheme, in which each MT can use any subcarrier at any time during the transmission, and TS-OFDMA may be seen as a special form of OFDMA. However, as mentioned in the previous sections, it is difficult to quantify the ``on'' period of each MT with the inequality in (\ref{eq:lower bound of on time}) being the only known expression. The proposed TS-OFDMA is thus more ``general'' than OFDMA and D-TDMA in the sense that it explicitly allows each MT to be off for a fraction of a frame (outside its assigned time slot) to save energy, and yet allows subcarriers sharing among users within the same time slot.}. An illustration of TS-OFDMA for a multiuser OFDM system with $K = 4$, $N = 4$, and $J = 3$ is given in Fig. \ref{fig:Schemeexample}(c).

\subsection{Solution to Problem (WSTREMin) with given $J$ and MT grouping}\label{sec:given group}
In this subsection, we solve problem (WSTREMin) based on TS-OFDMA with given $J$ and MT grouping. We first study two special cases, i.e., $J = K$ and $J = 1$, which can be regarded as the extensions of the results in Section \ref{sec:WSRE} and Section \ref{sec:TE}, respectively, by considering the weighted-sum transmitter and receiver energy consumption as the objective function. We thus have the following results.
\begin{enumerate}
\item $J = K$ and $|\Phi_j| = 1, j = 1,\cdots,J$: problem (WSTREMin) can be reformulated as
\begin{align}
\mathop{\mathtt{Min.}}\limits_{\{p_{k,n} \geq 0\}, \{t_k > 0\}} &
~~ \sum^{K}_{k=1}t_k\left(\alpha_kP_{r,c}+\alpha_0P_{t,c}\right) \nonumber \\
& + \alpha_0\sum^{K}_{k=1}t_k\sum_np_{k,n} \nonumber \\
\mathtt{s.t.}
& ~~ (\ref{eq:WSRT-TDMA c1}) ~ \text{and}~ (\ref{eq:WSRT-TDMA c2}). \label{eq:WSTREMin JK}
\end{align}
Note that for $J = K$, $T_k = t_k, \forall k$. Although problem (\ref{eq:WSTREMin JK}) and problem (WSREMin-TDMA) differ in their objective functions, problem (\ref{eq:WSTREMin JK}) can be recast as a convex problem similarly as problem (WSREMin-TDMA), and it can be shown that their optimal solutions possess the same structure. Therefore, problem (\ref{eq:WSTREMin JK}) can be solved by the algorithm similar to that in Table \ref{table1}.

\item $J = 1$ and $|\Phi_J| = K$: problem (WSTREMin) can be simplified to
\begin{align}
\mathop{\mathtt{Min.}}\limits_{\{p_{k,n}\}, \{\rho_{k,n}\},T} &
~~ \alpha_0T\sum^{K}_{k = 1}\sum^{N}_{n=1} \rho_{k,n}p_{k,n} \nonumber \\
&+T\left(\alpha_0P_{t,c}  + \sum^{K}_{k = 1}\alpha_kP_{r,c}\right) \nonumber \\
\mathtt{s.t.}
& ~~ (\ref{eq:WSTREMin c1}), (\ref{eq:WSTREMin c2}), (\ref{eq:WSTREMin c3}), ~ \text{and} ~(\ref{eq:WSTREMin c4}). \label{eq:WSTREMin J1}
\end{align}
Since problem (\ref{eq:WSTREMin J1}) has exactly the same structure as problem (TEMin), it can be solved by the algorithm similar to that in Table \ref{table3}.
\end{enumerate}

Next, consider the general case of $1 < J < K$. In this case, we divide $J$ slots into two sets as
\begin{align}
\mathcal{B}_1 & = \left\{j~:~ |\Phi_j|  =  1, j=1,\cdots,J\right\} \label{eq:group1}\\
\mathcal{B}_2 & = \left\{j~:~ |\Phi_j|  \geq 2, j=1,\cdots,J \right\} \label{eq:group2}
\end{align}
where $\mathcal{B}_1$ and $\mathcal{B}_2$ include slots that correspond to transmissions to single MT and multiple MTs, respectively. For slots in $\mathcal{B}_1$, we can further group them together and thereby formulate one single WSTREMin problem similarly as for the case of $J = K$. On the other hand, for slots in $\mathcal{B}_2$, we can perform WSTREMin in each slot separately similarly as for the case of $J = 1$. Furthermore, we assume that the average power assigned to all the slots in $\mathcal{B}_1$ and each slot in $\mathcal{B}_2$ are $P_{\text{avg}}$ to avoid coupled power allocation over these slots, so that each problem can be solved independently. Note that it is possible to jointly optimize the power allocation across all the slots. However, it requires extra complexity and thus this approach was not pursued.

The final tasks remaining in solving problem (WSTREMin) is to find the the optimal number of slots and to optimally assign MTs to each of these slots. Since finding the optimal grouping is a combinatorial problem, an exhaustive search can incur a large complexity if $K$ is large. To avoid the high complexity of exhaustive search, we propose a suboptimal MT grouping algorithm for $1 < J < K$ in Section \ref{sec:example strategy} next. The optimal $J$ can then be found by a one-dimension search.

\subsection{Suboptimal MT grouping algorithm for $1 <J < K$}\label{sec:example strategy}
In this subsection, we propose a suboptimal grouping algorithm for given $1 < J < K$, termed as \emph{channel orthogonality based grouping} (COG), with low complexity. The proposed algorithm is motivated by the observation that grouping MTs, whose strongest channels are orthogonal to each other (i.e. in different SCs), into one slot will not affect the power allocation and transmission time of each MT but will shorten the total transmission time, and thus reduce the total energy consumption.

For the purpose of illustration, we first define the following terms. Let $\mathbf{h}_{k} = [h_{k,1},\cdots,h_{k,N}]^T$ and $\hat{\mathbf{h}}_{k}$ denote the original and normalized (nonnegative) channel vector from the BS to MT $k$ across all SCs, respectively, where $\hat{\mathbf{h}}_{k} = \frac{\mathbf{h}_{k}}{\|\mathbf{h}_{k}\|}$. Furthermore, let $\pi_{k,l}$ denote the channel correlation index (CCI) between MTs $k$ and $l$, which is defined as the inner product between their normalized channel vectors, i.e.,
\begin{align}
\pi_{k,l} = \hat{\mathbf{h}}^T_{k}\hat{\mathbf{h}}_{l}, \forall k,l \neq k.
\end{align}
Note that $\pi_{k,l} = \pi_{l,k}$, and smaller (larger) $\pi_{k,l}$ indicates that MT $k$ is more (less) orthogonal to MT $j$ in terms of channel power realization across different SCs, which can be utilized as a cost associated with grouping MTs $k$ and $l$ into one slot. Finally, define the sum-CCI $\Pi_j$ of slot $j$ as
\begin{align}
\Pi_j = \sum_{l,k \in \Phi_j, l\neq k}\pi_{k,l}, j = 1,\cdots,J.
\end{align}

We are now ready to present the proposed COG algorithm for given $J$:
\begin{enumerate}
\item Compute the sum-CCI of MT $k$ to all other MTs, i.e. $\sum^{K}_{l \neq k}\pi_{k,l}$, $k = 1,\cdots, K$.
\item Assign the $J$ MTs corresponding to the first $J$ largest sum-CCI each to an individual time slot.
\item Each of the remaining $K - J$ MTs is successively assigned to one of the $J$ slots, which has the minimum increase of $\Pi_j$, $j=1,...,K$.
\end{enumerate}

\subsection{Algorithm for problem (WSTREMin)}
Combining the results in Section \ref{sec:given group} and Section \ref{sec:example strategy}, our complete algorithm for problem (WSTREMin) based on TS-OFDMA is summarized in Table \ref{table4}.

\begin{table}[ht]
\begin{center}
\caption{\textbf{Algorithm 3}: Algorithm for Solving Problem (WSTREMin)} \vspace{0.2cm}
 \hrule
\vspace{0.3cm}
\begin{enumerate}
\item Solve the two extreme cases, i.e. $J = K$ in (\ref{eq:WSTREMin JK}) and $J = 1$ in (\ref{eq:WSTREMin J1}), as described in Section \ref{sec:given group}.
\item For $1 < J < K$
    \begin{itemize}
        \item[ a)] Obtain the MT grouping by the COG algorithm.
        \item[ b)] Obtain $\mathcal{B}_1$ and $\mathcal{B}_2$ according to (\ref{eq:group1}) and (\ref{eq:group2}), respectively.
        \item[ c)] For slots in $\mathcal{B}_1$, solve one single WSTREMin problem similarly as for the case of $J = K$; for slots in $\mathcal{B}_2$, perform WSTREMin in each slot separately similarly as for the case of $J = 1$.
    \end{itemize}
\item Identify the optimal $J$ and MT grouping as the one resulting in the smallest WSTRE, and obtain its corresponding time and power allocations from the previous two steps.
\end{enumerate}
\vspace{0.2cm} \hrule \label{table4} \end{center}
\end{table}

Next, we analyze the complexity of the proposed algorithm in Table \ref{table4}. For step 1), the time complexity of the two extreme cases have been analyzed in Section \ref{sec:WSRE} and Section \ref{sec:TE}, which are of order $KN$ and $K^4+N^4+K^3N^3$, respectively. Therefore, the time complexity of step 1) is $\mathcal{O}(K^4+N^4+K^3N^3)$. For step 2), in each iteration with given $1 < J < K$, the computation time is dominated by solving separate WSTREMin problems for slots in $\mathcal{B}_1$ and $\mathcal{B}_2$ in step c), which depends on the MT grouping obtained by the COG algorithm. However, from the complexity analysis of the two extreme cases, it is observed that the worst case in terms of computation complexity is to assign as many as MTs into one slot, i.e., there are $J - 1$ slots in $\mathcal{B}_1$ but one slot in $\mathcal{B}_2$, which is of order $(K - J + 1)^4 + N^4 + (K-J+1)^3N^3$. Therefore, the overall worst case complexity of the algorithm in Table \ref{table4} is $\mathcal{O}(KN^4 + \sum^{K-1}_{J=1}(K-J+1)^4+(K-J+1)^3N^3)$, which is upper bounded by $\mathcal{O}(K^5+K^4N^3+KN^4)$.

\section{Time-constrained Optimization}\label{sec:delay}
{\color{red}We note that the total transmission time $T$ is practically bounded by $T \leq T_{\text{max}}$, where $T_{\text{max}}$ may be set as the channel coherence time or the maximum transmission delay constraint, whichever is smaller. In this section, we highlight the consequences of introducing the maximum transmission time constraint, and discuss in details how the obtained results in the previous sections can be extended to the case of time-constrained optimization.

Note that the maximum transmission time constraint, i.e., $T \leq T_{\text{max}}$, does not affect the solvability of problem (TEMin) in Section \ref{sec:TE} and problem (WSTREMin) under TS-OFDMA in Section \ref{sec:WSE}. However, in the case of maximum time constraint, the optimality of the TDMA structure for WSREMin may not hold in general (for example, when $\sum_{k=1}^K t_k^*> T_{\text{max}}$). However, Proposition \ref{proposition:1} reveals that orthogonalizing MTs' transmission in time is beneficial for WSREMin, which is useful even for the case of time-constrained optimization, since we may still assume D-TDMA structure to approximately solve problem (WSREMin). In the rest of this section, we discuss how to solve problems (WSREMin-TDMA), (TEMin) and (WSTREMin) under TS-OFDMA in the case with maximum time constraint $T_{\text{max}}$, which are termed as (WSREMin-TDMA-T), (TEMin-T) and (WSTREMin-T), respectively.

First, for problem (WSREMin-TDMA-T), it is observed that adding $\sum^{K}_{k=1}t_k \leq T_{\text{max}}$ does not affect its convexity after the same change of variables as problem (WSREMin-TDMA). Furthermore, the water-filling structure presented in Theorem \ref{theorem:1} still holds for problem (WSREMin-TDMA-T). Therefore, the algorithm in Table \ref{table1} can still be applied to solve problem (WSREMin-TDMA-D) with one additional step of bisection search over the maximum transmission time to ensure that it is no larger than $T_{\text{max}}$. On the other hand, the feasibility of problem (WSREMin-TDMA-T) can be verified by setting $\alpha_k = 1/P_{r,c},\forall k \in \mathcal{K}$ with the algorithm in Table \ref{table1}. If the obtained optimal value is smaller than $T_{\text{max}}$, it is feasible; otherwise, it is infeasible. It should be noted that, for the case with $T_{\text{max}}$, problem (WSTREMin-T) being feasible does not guarantee the feasibility of problem (WSREMin-TDMA-T) due to the prior assumed D-TDMA scheme.

For problem (TEMin-T), the maximum transmission time constraint does not affect its solvability compared with problem (TEMin), where the same decomposition method can be applied since problem (TEMin-2) with $T \leq T_{\text{max}}$ added, termed as (TEMin-2-T), is still convex. As a result, Lagrange duality method can again be applied to solve this problem optimally. Besides, the feasibility of (TEMin-T) can be checked by solving problem (TEMin-1) in Section \ref{sec:TE} with $T = T_{\text{max}}$ using the algorithm in Table \ref{table2}. If the obtained optimal value is smaller than the average power limit $P_{\text{avg}}$, problem (TEMin-T) is feasible; otherwise, it is infeasible.

Finally, for problem (WSTREMin-T) under TS-OFDMA with given $J$ and MT grouping (by the same suboptimal MT grouping algorithm as proposed in Section \ref{sec:example strategy}), time allocation needs to be optimized among different time slots to ensure the new maximum transmission time constraint. Since it can be shown that the optimal value of problem (WSREMin-TDMA-T) or (TEMin-T) is convex with respect to $T_{\text{max}}$ similarly as that in Lemma \ref{lemma:4}, gradient based method, e.g. Newton method \cite{Boydbook}, can be applied. Last, the feasibility of problem (WSTREMin-T) for given $J$ and MT grouping under TS-OFDMA can be checked by setting $\alpha_0 = 0$ and $\alpha_k = 1/P_{r,c}, \forall k \in \mathcal{K}$. If the obtained total transmission time is smaller than $T_{\text{max}}$, it is feasible; otherwise, it is infeasible.}

\section{Numerical Example}\label{sec:numerical}
In this section, we present simulation results to verify our theoretical analysis and demonstrate the tradeoffs in energy consumption at the BS and MTs. It is assumed that there are $K = 4$ active MTs with distances to the BS as $400$, $600$, $800$ and $700$ meters, and data requirements $\bar{Q}_k$ as $8.5$, $11.5$, $14.5$ and $17.5$ Kbits, respectively. The total number of SCs $N$ is set to be 16, and the bandwidth of each SC $W$ is $20$kHz. Independent multipath fading channels, each with six equal-energy independent consecutive time-domain taps, are assumed for each transmission link between each pair of the BS and MTs. Each tap coefficient consists of both small-scale fading and distance dependent attenuation components. The small-scale fading is assumed to be Rayleigh distributed with zero mean and unit variance, and the distance-dependent attenuation has a path-loss exponent equal to four. The power consumption of each MT, when turned on, is set to be $0.5$W. For the BS, we assume a constant non-transmission related power of $P_{t,c} = 20$W and an average transmit power of $P_{\text{avg}} = 30$W. We also set $\alpha_k = 1$ for all MTs, i.e., we consider the sum-energy consumption of all MTs. Finally, we set the receiver noise spectral density as $N_0=-174$ dBm/Hz, which corresponds to a typical thermal noise at room temperature.

\begin{figure}
\centering
\epsfxsize=0.7\linewidth
    \includegraphics[width=9cm]{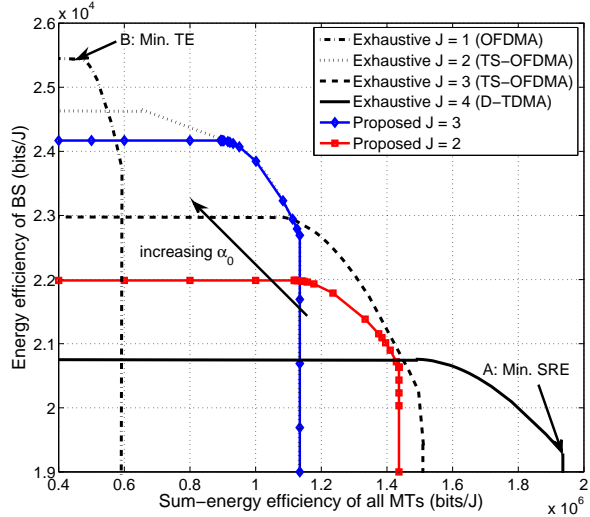}
\caption{Energy efficiency tradeoffs with different transmission schemes. The points ``Min. SRE'' and ``Min. TE'' represent the results obtained by methods in Section \ref{sec:WSRE} and Section \ref{sec:TE}, respectively.} \label{fig:Tradeoff}
\end{figure}

Fig. \ref{fig:Tradeoff} shows the energy efficiency tradeoffs (in bits/joule) between BS and MTs\footnote{For the ease of illustration, we treat the $K$ MTs as an ensemble, whose energy efficiency is defined as the ratio of sum-data received and sum-energy consumed at all MTs, i.e., $\sum^{K}_{k=1}\bar{Q}_k/\sum^{K}_{k=1}E_{r,k}$.} with various values of $J$, which is the number of orthogonal time slots in our proposed TS-OFDMA scheme in Section \ref{sec:WSE}, and by varying the value of the BS energy consumption weight $\alpha_0$ for each given $J$. In particular, the curves Exhaustive $J = 2$ and $J = 3$ are obtained by exhaustively searching all possible MT groupings, which serve as performance benchmark. The curves Proposed $J = 2$ and $J = 3$ are obtained by the COG algorithm presented in Section \ref{sec:example strategy}. The performance gap between the proposed algorithm and the benchmark is the price paid for lower computation complexity. It is observed that as $\alpha_0$ increases, the energy efficiency of BS increases and that at MTs decreases, respectively, for each $J$. It is easy to identify two boundary points of these tradeoff curves, namely, point A (on the curve of $J=4$ with $\alpha_0 = 0$) and B (on the curve of $J=1$ with $\alpha_0 = \infty$) correspond to the two special cases of TS-OFDMA, i.e., D-TDMA in Section \ref{sec:WSRE} and OFDMA in Section \ref{sec:TE}, respectively. By comparing the two boundary points, we observe that if BS's energy efficiency is reduced by $25\%$, then the sum-energy efficiency of MTs can be increased by around three times. Furthermore, it is observed that more flexible energy efficiency tradeoffs between BS and MTs than those in the cases of $J=1$ and $J=4$ can be achieved by applying the proposed TS-OFDMA transmission scheme with $J=2$ or $3$.

\begin{figure}
\centering
\epsfxsize=0.7\linewidth
    \includegraphics[width=9cm]{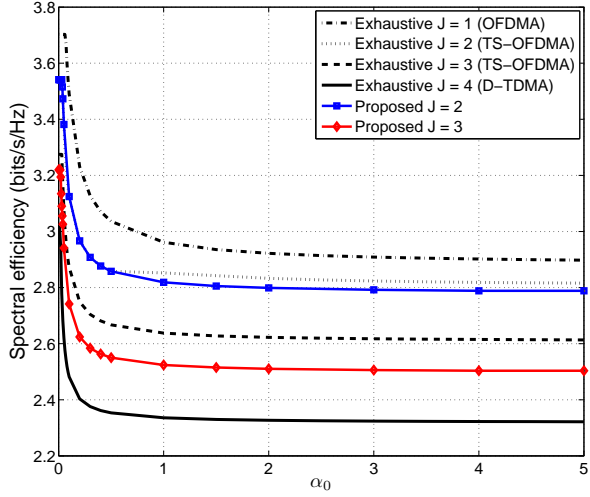}
\caption{Spectral efficiency comparison with different transmission schemes.} \label{fig:efficiency tradeoff}
\end{figure}

Next, in Fig. \ref{fig:efficiency tradeoff}, we show the spectral efficiency (in bits/s/Hz) of the considered multiuser downlink system over $\alpha_0$ with different values of $J$, which is defined as the total amount of transmitted data per unit time and bandwidth, i.e. $\sum^{K}_{k=1}\bar{Q}_k/TNW$. First, it is observed that the spectral efficiency decreases and finally converges as $\alpha_0$ increases for each value of $J$. The decreasing of spectral efficiency is the price to be paid for less energy consumption of BS (c.f. Fig. \ref{fig:Tradeoff}), which is due to the increase of the required transmission time $T$ and hence results in more energy consumption of MTs. It is also observed that for a given $\alpha_0$, the spectral efficiency decreases as $J$ increases, which is intuitively expected as $J = 1$, i.e., OFDMA, is known to be most spectrally efficient for multiuser downlink transmission.

\begin{remark}
For the proposed TS-OFDMA transmission scheme with given user grouping, the power consumption at MTs can be mathematically interpreted as extra non-transmission related power at the BS. As a result, problem (WSTREMin) can be treated as an equivalent merely transmitter-side energy minimization problem. As $\alpha_0$ increases, with proper normalization, it can be verified that the effective non-transmission related power decreases. Therefore, the optimal (most energy-efficient) transmission time $T$ will increase \cite{Bormann08}, which results in the decreasing of the spectral efficiency as shown in Fig. \ref{fig:efficiency tradeoff}. Furthermore, as $J$ increases (less MTs in each slot), MTs have more opportunity to be in the ``off'' mode to save energy, while on the contrary, BS has less opportunity to gain from so-called multiuser diversity \cite{Tse} to improve spectral efficiency. Consequently, the results in Fig. \ref{fig:Tradeoff} and Fig. \ref{fig:efficiency tradeoff} are expected.
\end{remark}

\section{Conclusion}\label{sec:conclusion}
In this paper, for cellular systems under an OFDM-based downlink communication setup, we have characterized the tradeoffs in minimizing the BS's versus MTs' energy consumptions by investigating a weighted-sum transmitter and receiver joint energy minimization (WSTREMin) problem, subject to an average transmit power constraint at the BS and data requirements of individual MTs. Two extreme cases, i.e., weighted-sum receiver-side energy minimization (WSREMin) for the MTs and transmitter-side energy minimization (TEMin) for the BS, are first solved separately. It is shown that Dynamic TDMA (D-TDMA) is the optimal transmission strategy for WSREMin, while OFDMA is optimal for TEMin. Based on the obtained resource allocation solutions in these two cases, we proposed a new multiple access scheme termed Time-Slotted OFDMA (TS-OFDMA) transmission scheme, which includes D-TDMA and OFDMA as special cases, to achieve more flexible energy consumption tradeoffs between the BS and MTs. The results of this paper provide important new insights to the optimal design of next generation cellular networks with their challenging requirements on both the spectral and energy efficiency of the network.

\appendices
\section{Proof of Proposition \ref{proposition:1}}\label{appendix:proof proposition 1}
We prove this proposition by contradiction. Suppose that $\left\{\{\rho^{a}_{k,n}\},\{p^{a}_{k,n}\},T^{a}\right\}$ (termed Solution A) is the optimal solution of problem (WSREMin), which does not satisfy the condition given in Proposition \ref{proposition:1}, i.e., there exists at least one MT $k$, such that its associated time sharing factors are not all identical. Next, we construct a new solution $\left\{\{\rho^{b}_{k,n}\},\{p^{b}_{k,n}\},T^{b}\right\}$ (termed Solution B) for problem (WSREMin), which satisfies the condition given in Proposition \ref{proposition:1} and also achieves a weighted-sum receiver energy consumption no larger than that of Solution A. The details of constructing Solution B are given as follows:
\begin{align}
T^{b} & = \sum^{K}_{k=1}\max_n\{\rho^{a}_{k,n}T^{a}\} \label{eq:construct 1}\\
\rho^{b}_{k,n} & = \max_j\{\rho^{a}_{k,j}\}/\sum^{K}_{i=1}\max_j\{\rho^{a}_{i,j}\},  \forall k, n \label{eq:construct 2}\\
r^{b}_{k,n} & = \left\{ \begin{array}{cl} \displaystyle
r^{a}_{k,n}\rho^{a}_{k,n}T^{a}/\rho^{b}_{k,n}T^{b} & \mbox{if } \rho^{a}_{k,n} > 0 \\
0 & \mbox{otherwise}. \end{array}\right. \forall k, n. \label{eq:construct 3}
\end{align}
Note that $\rho^{b}_{k,n} = \rho^{b}_{k}, \forall n \in \mathcal{N}, k \in \mathcal{K}$.

Next, we check that Solution B is also feasible for problem (WSREMin). Since
\begin{align}
\sum^{K}_{k=1}\rho^{b}_{k,n} & =  \sum^{K}_{k=1}\max_j\{\rho^{a}_{k,j}\}/\sum^{K}_{i=1}\max_j\{\rho^{a}_{i,j}\} = 1, \forall n \\
\sum^{N}_{n=1} T^{b}\rho^{b}_{k,n}r^{b}_{k,n} & = \sum^{N}_{n=1} T^{a}\rho^{a}_{k,n}r^{a}_{k,n} \geq \bar{Q}_k, \forall k
\end{align}
we verify that both the constraints in (\ref{eq:WSTREMin c1}) and (\ref{eq:WSTREMin c2}) are satisfied. Moreover, from (\ref{eq:construct 1}) and (\ref{eq:construct 2}), it is observed  that
\begin{align}
T^{b}\rho^{b}_{k,n} = T^{a}\max_j\{\rho^{a}_{k,j}\} \geq T^{a}\rho^{a}_{k,n}, \forall k, n
\end{align}
i.e., the time allocated to MT $k$ on SC $n$ in Solution B is no smaller than that in Solution A. From (\ref{eq:construct 3}), we also have that
\begin{align}
\rho^{b}_{k,n}T^{b}r^{b}_{k,n} = \rho^{a}_{k,n}T^{a}r^{a}_{k,n}, \forall k ,n
\end{align}
i.e., the amount of data delivered to MT $k$ on SC $n$ is the same for both Solution A and Solution B. Since $r_{k,n}$ is a strictly concave and increasing function of $p_{k,n}$, it is easy to verify that the amount of energy consumed for delivering the same amount of data decreases as the transmission time increases. Therefore, we have
\begin{align}
\sum^{K}_{k=1}\sum^{N}_{n=1} \rho^{b}_{k,n}p^{b}_{k,n} \leq \sum^{K}_{k=1}\sum^{N}_{n=1} \rho^{a}_{k,n}p^{a}_{k,n} \leq P_{\text{avg}}
\end{align}
i.e., the constraint in (\ref{eq:WSTREMin c3}) is satisfied by Solution B.

Finally, we show that Solution B achieves a weighted-sum receiver-side energy that is no larger than that by Solution A as follows. According to (\ref{eq:lower bound of on time}), we infer that
\begin{align}
t^{a}_k \geq \max\limits_n \{T^{a}\rho^{a}_{k,n}\}, \forall k
\end{align}
where $t^{a}_k$ is the on time of MT $k$ in Solution A. Let $t^{b}_k$ denote the on time of MT $k$ in Solution B. Since $\rho^{b}_{k,n}$'s are identical for given MT $k$, we can find $t^{b}_k$'s such that
\begin{align}
t^{b}_k = T^{b}\rho^{b}_{k,n} = \max_n\{T^{a}\rho^{a}_{k,n}\} \leq t^{a}_k, \forall k
\end{align}
which indicates that Solution B achieves a weighted-sum receiver-side energy no larger than that by Solution A. Thus, Proposition \ref{proposition:1} is proved.

\section{Proof of Theorem \ref{theorem:1}}\label{appendix:proof theorem 1}
Denote $\{s^{*}_{k,n}\}$ and $\{t^{*}_k\}$ as the optimal solution of problem (P1). Let $\beta$ and $\boldsymbol\lambda = [\lambda_1,\lambda_2,\cdots,\lambda_K]$ be the dual variables of problem (P1) associated with the average transmit power constraint in (\ref{eq:P1 c2}) and the data requirements in (\ref{eq:P1 c1}), respectively. Then the Lagrangian of problem (P1) can be expressed as
\begin{align}
& \mathcal{L}^{\text{P1}}(\{s_{k,n}\},\{t_k\},\boldsymbol\lambda, \beta) \nonumber \\
& = \sum^{K}_{k=1}\alpha_kP_{r,c}t_k -\sum^{K}_{k=1}\lambda_k\left(\sum^{N}_{n=1} s_{k,n} - \bar{Q}_k\right) \nonumber \\ &+ \beta\left(\sum^{K}_{k=1}\sum^{N}_{n=1}t_k\frac{e^{a\frac{s_{k,n}}{t_k}}-1}{f_{k,n}} - P_{\text{avg}}\sum^{K}_{k=1}t_k\right) \\
& = \sum^{K}_{k=1}\left(\alpha_kP_{r,c}t_k + \beta\sum^{N}_{n=1}t_k\frac{e^{a\frac{s_{k,n}}{t_k}}-1}{f_{k,n}} - \beta P_{\text{avg}}t_k\right) \nonumber \\
&-\sum^{K}_{k=1}\lambda_k\sum^{N}_{n=1} s_{k,n}+ \sum^{K}_{k=1}\lambda_k\bar{Q}_k. \label{eq:lagrangian of P1}
\end{align}

The Lagrange dual function of $\mathcal{L}^{\text{P1}}(\cdot)$ in (\ref{eq:lagrangian of P1}) is defined as
\begin{align}\label{eq:P1 dual function}
g^{\text{P1}}(\boldsymbol\lambda, \beta) = \mathop{\mathtt{Min.}}\limits_{\{s_{k,n} \geq 0\}, \{t_k > 0\}} \mathcal{L}^{\text{P1}}(\{s_{k,n}\},\{t_k\},\boldsymbol\lambda, \beta).
\end{align}
The dual problem of problem (P1) is expressed as
\begin{align}
\mathrm{(P1-D)}:~\mathop{\mathtt{Max.}}\limits_{\boldsymbol\lambda \geq 0,\beta \geq 0} &
~~ g^{\text{P1}}(\boldsymbol\lambda, \boldsymbol\beta).
\end{align}
Since (P1) is convex and satisfies the Salter's condition \cite{Boydbook}, strong duality holds between problem (P1) and its dual problem (P1-D). Let $\boldsymbol\lambda^{*} \geq 0$ and $\beta^{*} \geq 0$ denote the optimal dual solutions to problem (P1); then we have the following lemma.

\begin{lemma}\label{lemma:1}
The optimal solution to problem (P1-D) satisfies that
\begin{align}
\lambda^{*}_k &> 0, \forall k \label{eq:lemma11}\\
\beta^{*} &> 0\\
\beta^{*} - \min\limits_k(\alpha_k)P_{r,c}/P_{\text{avg}} & < 0 \label{eq:lemma12}\\
\alpha_kP_{r,c} - \beta^{*} P_{\text{avg}} + \sum^{N}_{n=1}u_{n}(\beta^{*},\lambda^{*}_k) &= 0, \forall k \label{eq:lemma13}
\end{align}
where $u_{n}(\beta,\lambda_k) = \left(\frac{\lambda_k}{a} - \frac{\beta}{f_{k,n}}\right)^{+} - \frac{\lambda_k}{a}\left(\ln\frac{\lambda_kf_{k,n}}{a\beta}\right)^{+}$ and $(\cdot)^{+} \triangleq \max\{\cdot,0\}$.
\end{lemma}
\begin{proof}
From (\ref{eq:lagrangian of P1}), it follows that the minimization of $\mathcal{L}^{\text{P1}}(\{s_{k,n}\},\{t_k\},\boldsymbol\lambda, \beta)$ can be decomposed into $K$ independent optimization problems, each for one MT and given by
\begin{align}\label{eq:decomposed problem of P1}
\mathop{\mathtt{Min.}}\limits_{t_k > 0, \{s_{k,n} \geq 0\}} ~ \mathcal{L}^{\text{P1}}_k(\{s_{k,n}\}, t_k, \lambda, \beta), ~ k=1,\cdots,K
\end{align}
where $\mathcal{L}^{\text{P1}}_k(\{s_{k,n}\}, t_k, \lambda, \beta) \triangleq \alpha_kt_k - \lambda_k\sum^{N}_{n=1} s_{k,n} + \beta\sum^{N}_{n=1}t_k\frac{e^{a\frac{s_{k,n}}{t_k}}-1}{f_{k,n}} - \beta P_{\text{avg}}t_k$. Note that $\mathcal{L}^{\text{P1}}(\{s_{k,n}\},\{t_k\},\boldsymbol\lambda, \beta) = \sum^{K}_{k=1}\mathcal{L}^{\text{P1}}_k(\cdot) + \sum^{K}_{k=1}\lambda_k\bar{Q}_k$.
By taking the derivative of $\mathcal{L}^{\text{P1}}_k(\cdot)$ with respect to $s_{k,n}$, we have
\begin{align}\label{eq:kkt1}
\frac{\partial \mathcal{L}^{\text{P1}}_k}{\partial s_{k,n}} = \frac{a\beta}{f_{k,n}}e^{a\frac{s_{k,n}}{t_k}} - \lambda_k.
\end{align}
Let $\{s^{\star}_{k,n}(\lambda_k, \beta)\}$ and $t^{\star}_k(\lambda_k, \beta)$ denote the optimal solution of problem (\ref{eq:decomposed problem of P1}) given $\lambda_k$ and $\beta$.

Next, we show that $\beta^{*} > 0$ and $\lambda^{*}_k > 0, \forall k$ by contradiction. If $\beta^{*} = 0$ and $\lambda^{*}_k = 0, \forall k$, from (\ref{eq:decomposed problem of P1}), it follows that $g^{\text{P1}}(\boldsymbol\lambda^{*}, \beta^{*}) = 0$, which is approached as $t_k \rightarrow 0, \forall k$, and the optimal value of problem (P1-D) is thus $0$, which contradicts with the fact that strong duality holds between problems (P1) and (P1-D). If $\beta^{*} = 0$ and $\exists i \in \mathcal{K}$ such that $\lambda^{*}_i > 0$, it follows that $\frac{\partial \mathcal{L}^{\text{P1}}_k}{\partial s_{i,n}} < 0, \forall n$ at the optimal dual solution, which implies that $s^{*}_{i,n} = \infty, \forall n$. Since $s_{i,n} = t_ir_{i,n}$, which is the amount of data delivered to MT $i$ on SC $n$ over the transmission, $s^{*}_{i,n} = \infty, \forall n$ indicates that $Q_i = \infty$, which is evidently suboptimal for problem (P1). If $\exists j \in \mathcal{K}$ such that $\lambda^{*}_j = 0$ and $\beta^{*} > 0$, it follows that $\frac{\partial \mathcal{L}^{\text{P1}}_j}{\partial s_{j,n}} > 0, \forall n$ at the optimal dual solution, which implies that $s^{*}_{j,n} = 0, \forall n$ or $Q_j = 0$. Then it contradicts with the fact that $\bar{Q}_j > 0$. Combining all the three cases above, it concludes that $\beta^{*} > 0$ and $\lambda^{*}_k > 0, \forall k$.

With $\beta > 0$ and $\lambda_k > 0, \forall k$ as proved above and from (\ref{eq:kkt1}), the ratio $\frac{s^{\star}_{k,n}(\lambda_k,\beta)}{t^{\star}_k(\lambda_k,\beta)}$ thus needs to satisfy
\begin{align}\label{eq:derivation P1}
\frac{s^{\star}_{k,n}(\lambda_k,\beta)}{t^{\star}_k(\lambda_k,\beta)} = \frac{1}{a}\left(\ln\frac{\lambda_kf_{k,n}}{a\beta}\right)^{+}, \forall n.
\end{align}
Substituting (\ref{eq:derivation P1}) back to $\mathcal{L}^{\text{P1}}_k(\cdot)$ yields
\begin{align}\label{eq:linear1}
\mathcal{L}^{\text{P1}}_{k}(\{s_{k,n}\}, t_k, \lambda, \beta) = \left(\alpha_kP_{r,c} - \beta P_{\text{avg}} + \sum^{N}_{n=1}u_{n}(\beta,\lambda_k)\right)t_k
\end{align}
which is a linear function of $t_k$ and thus $t^{*}_k$ is finite only if $\alpha_kP_{r,c} - \beta^{*} P_{\text{avg}} + \sum^{N}_{n=1}u_{n}(\beta^{*},\lambda^{*}_k) = 0$. Condition (\ref{eq:lemma13}) is thus verified.

Finally, we show that $\beta^{*} < \alpha_kP_{r,c}/P_{\text{avg}}$. Since it can be shown that given $\beta$, $\sum^{N}_{n=1}u_{n}(\beta,\lambda_k)$ equals zero when $\lambda_k \leq \frac{a\beta}{\max_n\{f_{k,n}\}}$ and is a strictly decreasing function of $\lambda_k$ when $\lambda_k > \frac{a\beta}{\max_n\{f_{k,n}\}}$, we have $\beta^{*} \leq \alpha_kP_{r,c}/P_{\text{avg}}$ from (\ref{eq:lemma11}). If $\beta^{*} = \alpha_kP_{r,c}/P_{\text{avg}}$, it follows that $\lambda^{*}_k \leq \frac{a\beta^{*}}{\max_n\{f_{k,n}\}}$, which implies that $s^{*}_{k,n} = 0, \forall n$ from (\ref{eq:derivation P1}). This again contradicts with the fact that $\bar{Q}_k > 0$. Lemma \ref{lemma:1} is thus proved.
\end{proof}

Next, we proceed to show the structural property of the optimal solution to problem (WSREMin-TDMA). Let the optimal solution of this problem be given by $\{p^{*}_{k,n}\}$ and $\{t^{*}_k\}$ with $s^{*}_{k,n} = r^{*}_{k,n}t^{*}_k, \forall n,k$, as in problem (P1). From the change of variables and (\ref{eq:rate on each subcarrier}), it follows that
\begin{align}\label{eq:t11}
\frac{s^{*}_{k,n}}{t^{*}_k} = W\log_2\left(1 + f_{k,n}p^{*}_{k,n}\right), \forall n, k.
\end{align}
Furthermore, from (\ref{eq:derivation P1}) we have
\begin{align}\label{eq:t12}
\frac{s^{*}_{k,n}}{t^{*}_k} = \frac{1}{a}\left(\ln\frac{\lambda^{*}_kf_{k,n}}{a\beta}\right)^{+}, \forall n, k.
\end{align}
Combining (\ref{eq:t11}) and (\ref{eq:t12}), (\ref{eq:P1 sol1}) can be easily verified.

From Lemma \ref{lemma:1} and the complementary slackness conditions \cite{Boydbook} satisfied by the optimal solution of problem (P1), it follows that
\begin{align}
\sum^{N}_{n=1} s^{*}_{k,n} & = \bar{Q}_k, \forall k \label{eq:data tight}\\
\sum^{K}_{k=1}\sum^{N}_{n=1}t^{*}_k\frac{e^{a\frac{s^{*}_{k,n}}{t^{*}_k}}-1}{f_{k,n}} & =  P_{\text{avg}}\sum^{K}_{k=1}t^{*}_k. \label{eq:power tight}
\end{align}
In other words, the optimal solutions of problem (P1) or problem (WSREMin-TDMA) are always attained with all the data constraints in (\ref{eq:P1 c1}) or (\ref{eq:WSRT-TDMA c1}) and average power constraint in (\ref{eq:P1 c2}) or (\ref{eq:WSRT-TDMA c2}) being met with equality. Substituting (\ref{eq:t12}) into (\ref{eq:data tight}), (\ref{eq:P1 sol2}) then easily follows. Theorem \ref{theorem:1} is thus proved.

\section{Solution to problem (P2)}\label{appendix:solution to p2}
The Lagrangian of problem (P2) can be expressed as
\begin{align}
& \mathcal{L}^{\text{P2}}(\{m_{k,n}\},\{\rho_{k,n}\},\boldsymbol\lambda, \boldsymbol\beta) \nonumber \\
& = \sum^{K}_{k=1}\sum^{N}_{n=1} \rho_{k,n}\frac{e^{a\frac{m_{k,n}}{\rho_{k,n}}}-1}{f_{k,n}} - \sum^{K}_{k=1}\lambda_k\left(\sum^{N}_{n=1} m_{k,n}-c_k\right)\nonumber \\
& + \sum^{N}_{n=1}\beta_n\left(\sum^{K}_{k=1}\rho_{k,n}-1\right) \\
& = \sum^{K}_{k=1}\sum^{N}_{n=1}\left(\rho_{k,n}\frac{e^{a\frac{m_{k,n}}{\rho_{k,n}}}-1}{f_{k,n}} - \lambda_km_{k,n}+\beta_n\rho_{k,n}\right)\nonumber \\
&+\sum^{K}_{k=1}\lambda_kc_k-\sum^{N}_{n=1}\beta_n \label{eq:lagrangian of P2}
\end{align}
where $\boldsymbol\lambda = [\lambda_1,\lambda_2,\cdots,\lambda_K]$ and $\boldsymbol\beta = [\beta_1,\beta_2,\cdots,\beta_N]$ are the vectors of dual variables associated with the constraints in (\ref{eq:P2 c2}) and (\ref{eq:P2 c1}), respectively.

Then, the corresponding dual function is defined as
\begin{align}\label{eq:P2 dual function}
g^{\text{P2}}(\boldsymbol\lambda, \boldsymbol\beta) = \mathop{\mathtt{Min.}}\limits_{\{m_{k,n} \geq 0\}, \{0 \leq \rho_{k,n} \leq 1\}} \mathcal{L}^{\text{P2}}(\{m_{k,n}\},\{\rho_{k,n}\},\boldsymbol\lambda, \boldsymbol\beta).
\end{align}
The dual problem of problem (P2) is thus expressed as
\begin{align}
\mathrm{(P2-D)}:~\mathop{\mathtt{Max.}}\limits_{\boldsymbol\lambda \geq 0,\boldsymbol\beta \geq 0} &
~~ g^{\text{P2}}(\boldsymbol\lambda, \boldsymbol\beta).
\end{align}
Since (P2) is convex and satisfies the Salter's condition \cite{Boydbook}, strong duality holds between problem (P2) and its dual problem (P2-D). To solve (P2-D), in the following we first solve problem (\ref{eq:P2 dual function}) to obtain $g(\boldsymbol\lambda, \boldsymbol\beta)$ with given $\boldsymbol\lambda \geq 0$ and $\boldsymbol\beta \geq 0$.

The expression of (\ref{eq:lagrangian of P2}) suggests that the minimization of $\mathcal{L}^{\text{P2}}(\{m_{k,n}\},\{\rho_{k,n}\},\boldsymbol\lambda, \boldsymbol\beta)$ can be decomposed into $NK$ parallel subproblems, each of which is for one given pair of $n$ and $k$ and expressed as
\begin{align}\label{eq:P2 decomposed problem}
\mathop{\mathtt{Min.}}\limits_{m_{k,n} \geq 0, 0 \leq \rho_{k,n} \leq 1}~ \mathcal{L}^{\text{P2}}_{k,n}(m_{k,n},\rho_{k,n},\lambda_k, \beta_n)
\end{align}
where $\mathcal{L}^{\text{P2}}_{k,n}(m_{k,n},\rho_{k,n},\lambda_k, \beta_n) \triangleq \rho_{k,n}\frac{e^{a\frac{m_{k,n}}{\rho_{k,n}}}-1}{f_{k,n}} - \lambda_km_{k,n}+\beta_n\rho_{k,n}$. Note that $\mathcal{L}^{\text{P2}}(\cdot) = \sum^{K}_{k=1}\sum^{N}_{n=1}\mathcal{L}^{\text{P2}}_{k,n}(\cdot) +\sum^{K}_{k=1}\lambda_kc_k-\sum^{N}_{n=1}\beta_n$.

\begin{lemma}\label{lemma:6}
The optimal solution of problem (P2-D) satisfies that $\boldsymbol\lambda^{*} > 0$ and $\boldsymbol\beta^{*} > 0$.
\end{lemma}
\begin{proof}
The proof is similar to that of Lemma \ref{lemma:1}, and thus is omitted for brevity.
\end{proof}
With Lemma \ref{lemma:6}, in the following, we only consider the case that $\boldsymbol\lambda > 0$ and $\boldsymbol\beta > 0$.
\begin{lemma}\label{lemma:3}
For a given pair of $n$ and $k$ with $\lambda_k > 0$ and $\beta_n > 0$, the optimal solution of problem (\ref{eq:P2 decomposed problem}) is given by
\begin{align}
m^{\star}_{k,n}(\lambda_k,\beta_n) & = \frac{\rho^{\star}_{k,n}(\lambda_k,\beta_n)}{a} \left(\ln\frac{\lambda_kf_{k,n}}{a}\right)^{+} \label{eq:P2 solution 1}\\
\rho^{\star}_{k,n}(\lambda_k,\beta_n) & = \left\{ \begin{array}{cl} \displaystyle
1 & o(\lambda_k, \beta_n) < 0 \\
0 & \mbox{otherwise }
\end{array}\right.\label{eq:P2 solution 2}
\end{align}
where $o(\lambda_k, \beta_n) = \left(\frac{\lambda_k}{a} - \frac{1}{f_{k,n}}\right)^{+} - \frac{\lambda_k}{a}\left(\ln\frac{\lambda_kf_{k,n}}{a}\right)^{+} + \beta_n$.
\end{lemma}
\begin{proof}
First, consider the case of $\rho_{k,n} = 0$, in which $m_{k,n} = 0$ and $\frac{m_{k,n}}{\rho_{k,n}} = 0$. It follows that $\mathcal{L}^{\text{P2}}_{k,n}(\cdot) = 0$.

Second, consider the case of $\rho_{k,n} > 0$. Taking the derivative of $\mathcal{L}^{\text{P2}}_{k,n}(\cdot)$ over $m_{k,n}$ and $\rho_{k,n}$, respectively, we have
\begin{align}
\frac{\partial \mathcal{L}^{\text{P2}}_{k,n}}{\partial m_{k,n}} & = \frac{a}{f_{k,n}}e^{a\frac{m_{k,n}}{\rho_{k,n}}} - \lambda_k \label{eq:derivative 1}\\
\frac{\partial \mathcal{L}^{\text{P2}}_{k,n}}{\partial \rho_{k,n}} & = \frac{1}{f_{k,n}}e^{a\frac{m_{k,n}}{\rho_{k,n}}}\left(1 - a \frac{m_{k,n}}{\rho_{k,n}}\right) - \frac{1}{f_{k,n}} + \beta_n. \label{eq:derivative 2}
\end{align}
Then it is easy to see that given $\lambda_k > 0$ and $\beta_n > 0$, from (\ref{eq:derivative 1}), the optimal solution of problem (\ref{eq:P2 decomposed problem}) needs to satisfy the following equation:
\begin{align}\label{eq:P2 optimal condition}
m^{\star}_{k,n}(\lambda_k,\beta_n)= \frac{\rho^{\star}_{k,n}(\lambda_k,\beta_n) }{a}\left(\ln\frac{\lambda_kf_{k,n}}{a}\right)^{+}.
\end{align}
Substituting (\ref{eq:P2 optimal condition}) into (\ref{eq:derivative 2}), it then follows that $\frac{\partial \mathcal{L}^{\text{P2}}_{k,n}}{\partial \rho_{k,n}} = o(\lambda_k, \beta_n)$, which is a constant implying
\begin{align}\label{eq:solution for p2 derivative}
\rho^{\star}_{k,n}(\lambda_k,\beta_n) & = \left\{ \begin{array}{cl} \displaystyle
1 & \mbox{if } o(\lambda_k, \beta_n) < 0 \\
\left(0,1\right] & \mbox{if } o(\lambda_k, \beta_n) = 0 \\
\rightarrow 0 & \mbox{otherwise }
\end{array}\right.
\end{align}
where $\rightarrow 0$ means here that the optimal value cannot be attained but can be approached as $\rho^{\star}_{k,n}(\lambda_k,\beta_n) \rightarrow 0$. Then, substituting (\ref{eq:P2 optimal condition}) into $\mathcal{L}^{\text{P2}}_{k,n}(\cdot)$, it follows that $\mathcal{L}^{\text{P2}}_{k,n}(\cdot) = \rho^{\star}_{k,n}(\lambda_k,\beta_n)o(\lambda_k, \beta_n)$. Thus, (\ref{eq:solution for p2 derivative}) achieves the optimal value of $\mathcal{L}^{\text{P2}}_{k,n}(\cdot)$ as
\begin{align}
\mathcal{L}^{\text{P2}}_{k,n}(\cdot) & = \left\{ \begin{array}{cl} \displaystyle
o(\lambda_k, \beta_n) & \mbox{if } o(\lambda_k, \beta_n) < 0 \\
0 & \mbox{otherwise. }
\end{array}\right.
\end{align}

Combining the two cases above, Lemma \ref{lemma:3} is thus proved.
\end{proof}

With Lemma \ref{lemma:3}, we can solve the $NK$ subproblems in (\ref{eq:P2 decomposed problem}) and thus obtain $g(\boldsymbol\lambda, \boldsymbol\beta)$ with given $\boldsymbol\lambda > 0$ and $\boldsymbol\beta > 0$. Then, we solve problem (P2-D) by finding the optimal $\boldsymbol\lambda$ and $\boldsymbol\beta$ to maximize $g(\boldsymbol\lambda, \boldsymbol\beta)$. Although problem (P2-D) is convex, the dual function $g(\boldsymbol\lambda, \boldsymbol\beta)$ is not differentiable and as a result analytical expressions for its differentials do not exist. Hence, conventional methods with gradient based search, such as Newton method, cannot be applied for solving problem (P2-D). An alternative method is thus the ellipsoid method \cite{Boyd2}, which is capable of minimizing non-differentiable convex functions based on the so-called subgradient.\footnote{The subgradient of $g(\boldsymbol\lambda, \boldsymbol\beta)$ at given $\boldsymbol\lambda$ and $\boldsymbol\beta$ for the ellipsoid method can be shown to be $\sum_nm^{\star}_{k,n}(\lambda_k,\beta_n) - c_k$ for $\lambda_k$, $k=1,\cdots,K$ and $1 - \sum_k\rho^{\star}_{k,n}(\lambda_k,\beta_n)$ for $\beta_n$, $n = 1,\cdots,N$.}
Hence, the optimal solution of (P2-D) can be obtained as $\boldsymbol\lambda^{*}$ and $\boldsymbol\beta^{*}$ by applying the ellipsoid method.

After obtaining the dual solution $\boldsymbol\lambda^{*}$ and $\boldsymbol\beta^{*}$, we can substitute them into (\ref{eq:P2 solution 1}) and (\ref{eq:P2 solution 2}), and obtain the corresponding $\{m^{\star}_{k,n}\}$ and $\{\rho^{\star}_{k,n}\}$. However, notice that the obtained $\{m^{\star}_{k,n}\}$ and $\{\rho^{\star}_{k,n}\}$ may not necessarily be the optimal solution of problem (P2), denoted by $\{m^{*}_{k,n}\}$ and $\{\rho^{*}_{k,n}\}$, since they may not satisfy the constraints in (\ref{eq:P2 c1}) and (\ref{eq:P2 c2}). The reason is that when $o(\lambda^{*}_k, \beta^{*}_n) = 0$ for certain pairs of $n$ and $k$, the corresponding $\rho^{\star}_{k,n}$ can actually take any value within $[0,1]$ according to (\ref{eq:solution for p2 derivative}), each of which would result in a different $m^{\star}_{k,n}$ accordingly. Therefore, with $\boldsymbol\lambda^{*}$ and $\boldsymbol\beta^{*}$, we may obtain infinite sets of $\{m^{\star}_{k,n}\}$ and $\{\rho^{\star}_{k,n}\}$, some of which might not satisfy the constraints in (\ref{eq:P2 c1}) and/or (\ref{eq:P2 c2}) \cite{Rui06}. In such cases, a linear programming (LP) needs to be further solved to obtain a feasible optimal solution for problem (P2).

To be more specific, we first define the following two sets with given $\boldsymbol\lambda^{*}$ and $\boldsymbol\beta^{*}$:
\begin{align}
\mathcal{A}_1 & = \left\{(k,n)~|~ o(\lambda^{*}_k, \beta^{*}_n) \neq 0, \forall k,n\right\} \\
\mathcal{A}_2 & = \left\{(k,n)~|~ o(\lambda^{*}_k, \beta^{*}_n) = 0, \forall k,n\right\}.
\end{align}
From (\ref{eq:P2 solution 1}) and (\ref{eq:P2 solution 2}), we know that for any pair of $n$ and $k$ with $(k,n) \in \mathcal{A}_1$, the corresponding $m^{\star}_{k,n}$ and $\rho^{\star}_{k,n}$ can be uniquely determined, which implies
\begin{align}
m^{*}_{k,n} = m^{\star}_{k,n}, \rho^{*}_{k,n} = \rho^{\star}_{k,n}, \forall (k,n) \in \mathcal{A}_1.
\end{align}
The problem remains to find $m^{*}_{k,n}$ and $\rho^{*}_{k,n}$ with $(k,n) \in \mathcal{A}_2$. It is then observed that the optimal solution of problem (P2) needs to satisfy the following linear equations:
\begin{align}
m^{*}_{k,n} & = \frac{\rho^{*}_{k,n}}{a} \left(\ln\frac{\lambda^{*}_kf_{k,n}}{a}\right)^{+}, \forall k,n \label{eq:lp1}\\
\sum_{k}\rho^{*}_{k,n} &= 1, \forall n,~~  \sum_n m^{*}_{k,n} = c_k, \forall k \label{eq:lp2}
\end{align}
where (\ref{eq:lp1}) is due to (\ref{eq:P2 solution 1}), and (\ref{eq:lp2}) is due to Lemma \ref{lemma:6} and the complementary slackness conditions \cite{Boydbook} satisfied by the optimal solution of problem (P2). Therefore, $m^{*}_{k,n}$ and $\rho^{*}_{k,n}$ with $(k,n) \in \mathcal{A}_2$ can be found through solving the above linear equations by treating $m^{*}_{k,n}$ and $\rho^{*}_{k,n}$ with $(k,n) \in \mathcal{A}_1$ as given constants, which is a linear programming (LP) and can be efficiently solved. In summary, one algorithm for solving problem (P2) and its dual problem (P2-D) is given in Table \ref{table2} as follows.

For the algorithm given in Table \ref{table2}, the computation time is dominated by the ellipsoid method in steps 1)-3) and the LP in step 4). In particular, the time complexity of steps 1)-3) is of order $(K+N)^4$ \cite{Boyd2} , step 4) is of order $K^3N^3$ \cite{Boydbook}. Therefore, the time complexity of the algorithm in Table \ref{table2} is $\mathcal{O}(K^4+N^4+K^3N^3)$.

\begin{table}[ht]
\begin{center}
\caption{\textbf{Algorithm 2}: Algorithm for Solving Problem (P2) and (P2-D)} \vspace{0.2cm}
 \hrule
\vspace{0.3cm}
\begin{enumerate}
\item Initialize $\boldsymbol\lambda > \mathbf{0}$ and $\boldsymbol\beta > 0$.
\item {\bf Repeat:}
    \begin{itemize}
    \item[ a)] Obtain $\{m^{*}_{k,n}(\lambda_k,\beta_n)\}$ and $\{\rho^{*}_{k,n}(\lambda_k,\beta_n)\}$ using (\ref{eq:P2 solution 1}) and (\ref{eq:P2 solution 2}), respectively, with given $\boldsymbol\lambda$ and $\boldsymbol\beta$.
    \item[ b)] Compute the subgradient of $g(\boldsymbol\lambda, \boldsymbol\beta)$ and update $\boldsymbol\lambda$ and $\boldsymbol\beta$ accordingly using the ellipsoid method \cite{Boyd2}.
    \end{itemize}
\item {\bf Until }both $\boldsymbol\lambda$ and $\boldsymbol\beta$ converge to $\boldsymbol\lambda^{*}$ and $\boldsymbol\beta^{*}$, respectively, within a prescribed accuracy.
\item Determine $\{\{m^{\star}_{k,n}\}, \{\rho^{\star}_{k,n}\}\}$ with $\boldsymbol\lambda^{*}$ and $\boldsymbol\beta^{*}$. If it is feasible for problem (P2), set $\{\{m^{*}_{k,n}\}, \{\rho^{*}_{k,n}\}\} = \{\{m^{\star}_{k,n}\}, \{\rho^{\star}_{k,n}\}\}$; otherwise solve a LP to find $\{\{m^{*}_{k,n}\}, \{\rho^{*}_{k,n}\}\}$.
\end{enumerate}
\vspace{0.2cm} \hrule \label{table2} \end{center}
\end{table}

\section{Proof of Lemma \ref{lemma:4}}\label{appendix:proof lemma 4}
To show problem (TEMin-2) is convex, we need to prove that both $v(T)$ and $v(T)T+ P_{t,c}T$ are convex functions of $T$. Since $P_{t,c}T$ is linear in $T$, we only need to show the convexity of $v(T)$ and $v(T)T$.

First, we check the convexity of function $v(T)$, which is sufficient to prove that for any convex combination $T = \theta T_1 + (1-\theta)T_2$ with $T_1,T_2 > 0$ and $\theta \in (0,1)$, we have $v(T) \leq \theta v(T_1) + (1-\theta)v(T_2)$. Denote the optimal solution to problem (TEMin-1) with $T_1$ and $T_2$ as $\{\dot{p}^{*}_{k,n}\}$, $\{\dot{\rho}^{*}_{k,n}\}$ (termed Solution 1) and $\{\ddot{p}^{*}_{k,n}\}$, $\{\ddot{\rho}^{*}_{k,n}\}$ (termed Solution 2), respectively. Then we have
\begin{align}\label{eq:time sharing objective}
\theta v(T_1) + (1-\theta)v(T_2) & = \theta\sum^{K}_{k=1}\sum^{N}_{n=1}\dot{\rho}^{*}_{k,n}\dot{p}^{*}_{k,n} \nonumber \\
&+ (1-\theta)\sum^{K}_{k=1}\sum^{N}_{n=1}\ddot{\rho}^{*}_{k,n}\ddot{p}^{*}_{k,n}.
\end{align}

Next we construct another solution $\{\bar{p}^{*}_{k,n}\}$, $\{\bar{\rho}^{*}_{k,n}\}$ (termed Solution 3) of problem (TEMin-1) with given $T$, which is achieved by properly allocating power for each MT on each SC such that the average power consumption is the same as that with time sharing between Solution 1 and Solution 2. The details of constructing Solution 3 are given as follows:
\begin{align}
\bar{\rho}^{*}_{k,n} & = \frac{\dot{\rho}^{*}_{k,n}\theta T_1 + \ddot{\rho}^{*}_{k,n}(1-\theta)T_2}{\theta T_1 + (1-\theta)T_2} \\
\bar{p}^{*}_{k,n} & = \frac{\dot{p}^{*}_{k,n}\dot{\rho}^{*}_{k,n}\theta T_1 + \ddot{p}^{*}_{k,n}\ddot{\rho}^{*}_{k,n}(1-\theta)T_2}{\bar{\rho}^{*}_{k,n}[\theta T_1 + (1-\theta)T_2]}.
\end{align}
It can then be shown that
\begin{align}
\sum^{K}_{k=1}\bar{\rho}^{*}_{k,n} & = \frac{\theta T_1\sum^{K}_{k=1}\dot{\rho}^{*}_{k,n} + (1-\theta)T_2\sum^{K}_{k=1}\ddot{\rho}^{*}_{k,n}}{\theta T_1 + (1-\theta)T_2} \nonumber \\
& \leq \frac{\theta T_1 + (1-\theta)T_2}{\theta T_1 + (1-\theta)T_2} = 1 \\
\sum^{N}_{n=1}T\bar{\rho}^{*}_{k,n}\bar{r}^{*}_{k,n} & = \left(\theta T_1\sum^{N}_{n=1}\dot{\rho}^{*}_{k,n} + (1- \theta)T_2\sum^{N}_{n=1}\ddot{\rho}^{*}_{k,n}\right)\bar{r}^{*}_{k,n} \nonumber \\
& \geq \theta T_1\sum^{N}_{n=1}\dot{\rho}^{*}_{k,n}\dot{r}^{*}_{k,n} + (1- \theta)T_2\sum^{N}_{n=1}\ddot{\rho}^{*}_{k,n}\ddot{r}^{*}_{k,n} \nonumber \\
& \geq \theta\bar{Q}_k + (1-\theta)\bar{Q}_k = \bar{Q}_k \\
\sum^{K}_{k=1}\sum^{N}_{n=1}\bar{\rho}^{*}_{k,n}\bar{p}^{*}_{k,n} & = \frac{\theta T_1\sum^{K}_{k=1}\sum^{N}_{n=1}\dot{\rho}^{*}_{k,n}\dot{p}^{*}_{k,n}}{\theta T_1 + (1-\theta)T_2} \nonumber \\
& + \frac{(1-\theta)T_2\sum^{K}_{k=1}\sum^{N}_{n=1}\ddot{\rho}^{*}_{k,n}\ddot{p}^{*}_{k,n}}{\theta T_1 + (1-\theta)T_2} \nonumber \\
& \leq \sum^{K}_{k=1}\sum^{N}_{n=1}\dot{\rho}^{*}_{k,n}\dot{p}^{*}_{k,n} + \sum^{K}_{k=1}\sum^{N}_{n=1}\ddot{\rho}^{*}_{k,n}\ddot{p}^{*}_{k,n}
\end{align}
i.e., Solution 3 is feasible for problem (TEMin-1) with the given $T$, and also achieves the same objective value as that in (\ref{eq:time sharing objective}). Since Solution 3 is only a feasible solution for problem (TEMin-1) with given $T$, which is not necessary to be optimal, we have
\begin{align}
v(T) \leq \theta v(T_1) + (1-\theta)v(T_2).
\end{align}
The convexity of $v(T)$ is thus proved.

Similar arguments can be applied to verify the convexity of $v(T)T$; Lemma \ref{lemma:4} is thus proved.

\section{Proof of Lemma \ref{lemma:5}}\label{appendix:proof lemma 5}
First, we find the gradient of $v(T)$. Since $v(T)$ is differentiable, its gradient and subgradient are equivalent. We provide the definition of subgradient \cite{Boyd2} as follows. A vector $y \in \mathbf{R}^{n}$ is said to be the subgradient of function $q : \mathbf{R}^{n} \rightarrow \mathbf{R}$ at $x \in \mathbf{dom} ~ q$ if for all $z \in \mathbf{dom} ~ q$,
\begin{align}
q(z) \geq q(x) + y^{T}(z - x).
\end{align}

The dual function (\ref{eq:P2 dual function}) can be expressed as
\begin{align}
& g(\boldsymbol\lambda, \boldsymbol\beta) \nonumber \\
&= \inf\limits_{\{m_{k,n}\}, \{ \rho_{k,n}\}}\sum^{K}_{k=1}\sum^{N}_{n=1}\left(\rho_{k,n}\frac{e^{a\frac{m_{k,n}}{\rho_{k,n}}}-1}{f_{k,n}} - \lambda_km_{k,n}+\beta_n\rho_{k,n}\right)\nonumber \\
&+\frac{1}{T}\sum^{K}_{k=1}\lambda_k\bar{Q}_k-\sum^{N}_{n=1}\beta_n.
\end{align}
Then, we have
\begin{align}
& v(T) = \mathop{\mathtt{Max.}}\limits_{\boldsymbol\lambda \geq 0, \boldsymbol\beta \geq 0} g(\boldsymbol\lambda, \boldsymbol\beta) \\
& = \inf\limits_{\{m_{k,n}\}, \{\rho_{k,n}\}}\sum^{K}_{k=1}\sum^{N}_{n=1}\left(\rho_{k,n}\frac{e^{a\frac{m_{k,n}}{\rho_{k,n}}}-1}{f_{k,n}} - \lambda^{*}_k(T)m_{k,n}\right.\nonumber\\
&+\beta^{*}_n(T)\rho_{k,n}\Bigg)+\frac{1}{T}\sum^{K}_{k=1}\lambda^{*}_k(T)\bar{Q}_k-\sum^{N}_{n=1}\beta^{*}_n(T)
\end{align}
where $\{\lambda^{*}_k(T)\}$ and $\{\beta^{*}_n(T)\}$ is the optimal solution of problem (P2-D) with given $T > 0$. For any $T^{'} > 0$ and $T^{'} \neq T$, we have
\begin{align}
v(T^{'}) & = \mathop{\mathtt{Max.}}\limits_{\boldsymbol\lambda \geq 0, \boldsymbol\beta \geq 0} \inf\limits_{\{m_{k,n}\}, \{\rho_{k,n}\}}\sum^{K}_{k=1}\sum^{N}_{n=1}\left(\rho_{k,n}\frac{e^{a\frac{m_{k,n}}{\rho_{k,n}}}-1}{f_{k,n}}\right. \nonumber \\
& - \lambda_km_{k,n}+\beta_n\rho_{k,n}\Bigg)+\frac{1}{T^{'}}\sum^{K}_{k=1}\lambda_k\bar{Q}_k-\sum^{N}_{n=1}\beta_n \\
& \geq \inf\limits_{\{m_{k,n}\}, \{\rho_{k,n}\}}\sum^{K}_{k=1}\sum^{N}_{n=1}\left(\rho_{k,n}\frac{e^{a\frac{m_{k,n}}{\rho_{k,n}}}-1}{f_{k,n}} - \lambda^{*}_k(T)m_{k,n}\right.\nonumber \\
& +\beta^{*}_n(T)\rho_{k,n}\Bigg)+\frac{1}{T^{'}}\sum^{K}_{k=1}\lambda^{*}_k(T)\bar{Q}_k-\sum^{N}_{n=1}\beta^{*}_n(T) \\
& = v(T) + \left(\frac{1}{T^{'}} - \frac{1}{T}\right)\sum^{K}_{k=1}\lambda^{*}_k(T)\bar{Q}_k \\
& = v(T) + \left(-\frac{1}{T^2}\sum^{K}_{k=1}\lambda^{*}_k(T)\bar{Q}_k\right)(T - \frac{T^2}{T^{'}}) \\
& \geq v(T) + \left(-\frac{1}{T^2}\sum^{K}_{k=1}\lambda^{*}_k(T)\bar{Q}_k\right)(T^{'} - T)
\end{align}
where the last inequality is due to $\left(T - \frac{T^2}{T^{'}}\right) - (T^{'} - T) = (T^{'} - T)\left(\frac{T}{T^{'}} - 1\right) < 0$. Thus, the subgradient (gradient) of $v(T)$ is given by
\begin{align}
v^{'}(T) = -\frac{1}{T^2}\sum^{K}_{k=1}\lambda^{*}_k(T)\bar{Q}_k.
\end{align}

With the gradient of $v(T)$, Lemma \ref{lemma:5} can be easily verified.

\end{document}